\documentclass[11pt]{article}
\usepackage[english]{babel}
\usepackage[utf8]{inputenc}
\usepackage{lmodern}
\usepackage[T1]{fontenc}
\usepackage{amsmath,amssymb,lmodern,amsthm}
\usepackage{marginnote}
\usepackage{mathtools}
\usepackage{graphicx}
\usepackage{mathrsfs}
\usepackage{bm}
\usepackage{url}
\usepackage{breakurl}
\usepackage{array}
\usepackage{tensor}
\usepackage{pifont}
\usepackage[toc,page]{appendix}
\usepackage[left=2.5cm,right=2.5cm,top=3cm,bottom=3cm]{geometry}
\usepackage{tikz}
\usetikzlibrary{arrows,positioning}
\usepackage{enumitem}
\usepackage{empheq}

\usepackage{cases}
\usepackage{csquotes}
\usepackage{accents}
\usepackage{setspace}

\theoremstyle{plain}
\usepackage{xargs}
\usepackage{xcolor}
\usepackage{subcaption}
\usepackage[justification=raggedright]{caption}
\usepackage{hyperref}
\usepackage{cleveref}
\usepackage{tocloft}

\hypersetup{
	colorlinks=true,
	urlcolor=blue,
	linkcolor={black},
	citecolor={green!40!black}
}
\usepackage[affil-it]{authblk}
\usepackage{thmtools}
\usepackage[backend=bibtex,style=numeric-comp,maxnames=4,sorting=none,url=false,hyperref=true,date=year,eprint=true,doi=true]{biblatex}

\newlist{properties}{enumerate}{10}
\setlist[properties]{label*=\roman*)}
\crefname{propertiesi}{\text{property}}{\text{properties}}
\Crefname{propertiesi}{\text{Property}}{\text{Properties}}
\newlist{points}{enumerate}{10}
\setlist[points]{label*=\arabic*)}
\crefname{pointsi}{\text{point}}{\text{points}}
\Crefname{pointsi}{\text{Point}}{\text{Points}}
\newcommand{\ubar}[1]{\underaccent{\bar}{#1}}
\numberwithin{equation}{section}
\newtheorem{thm}{Theorem}[]
\theoremstyle{definition}
\newtheorem{prop}{Proposition}[section]
\newtheorem{corollary}{Corollary}[section]
\newtheorem{deff}{Definition}[section]
\newtheorem*{crit*}{Criterion}
\newtheorem{crit}{Criterion}[]

\newtheorem{lemma}{Lemma}[section]
\newtheorem{remark}{Remark}[section]
\renewcommand{\iff}{\Longleftrightarrow}

\newcommand{\lied}{\pounds}
\newcommand{\defeq}{\vcentcolon=}
\newcommand{\defeqs}{\stackrel{\scri}{\vcentcolon=}}
\newcommand{\defeqc}{\stackrel{\Sc}{\vcentcolon=}}

\newcommand{\scri}{\mathscr{J}}

\newcommand{\pt}[2]{\tensor{\hat{#1}}{#2}}

\newcommand{\ctr}[3]{\tensor[{#1}]{#2}{#3}}

\newcommand{\ct}[2]{\tensor{{#1}}{#2}}

\newcommand{\cts}[2]{\tensor{\overline{#1}}{#2}}

\newcommand{\ctd}[2]{\tensor{\dot{#1}}{#2}}

\newcommand{\cscn}[1]{\ubar{{#1}}}
\newcommand{\ctcn}[2]{\tensor{\ubar{{#1}}}{#2}}

\newcommand{\cdcn}[1]{\tensor{\ubar{\D}}{#1}}
\newcommand{\cdt}[1]{\tensor{\overline{\D}}{#1}}
\newcommand{\cdtp}[1]{\tensor{{\overline{\D}'}}{#1}}
\newcommand{\cdto}[1]{\tensor{{\mathring{\overline{\D}}}}{#1}}

\newcommand{\ctcng}[2]{\tensor{\tilde{\ubar{#1}}}{#2}}
\newcommand{\ctg}[2]{\tensor{\tilde{#1}}{#2}}

\newcommand{\ctsg}[2]{\tensor{\tilde{\overline{#1}}}{#2}}

\newcommand{\ctcna}[2]{\tensor{{\acute{\ubar{#1}}}}{#2}}
\newcommand{\cta}[2]{\tensor{\acute{#1}}{#2}}
\newcommand{\ctap}[2]{\tensor{{\acute{#1}'}}{#2}}

\newcommand{\cdcna}[1]{\tensor{\ubar{\acute{\D}}}{#1}}
\newcommand{\ctcnh}[2]{\tensor{\hat{\ubar{#1}}}{#2}}
\newcommand{\cth}[2]{\tensor{\hat{#1}}{#2}}
\newcommand{\cthp}[2]{\tensor{{\hat{#1}'}}{#2}}

\newcommand{\cdcnh}[1]{\tensor{\ubar{\hat{\D}}}{#1}}
\newcommand{\cto}[2]{\tensor{\mathring{#1}}{#2}}

\newcommand{\csS}[1]{\overline{#1}}

\newcommand{\df}[1]{\text{d}#1}
\newcommand{\prn}[1]{\left(#1\right)}

\newcommand{\brkt}[1]{\left[#1\right]}

\newcommand{\cbrkt}[1]{\left\lbrace#1\right\rbrace}

\newcommand{\eqs}{\stackrel{\scri}{=}}
\newcommand{\equp}{\stackrel{\U_{p}}{=}}

\newcommand{\eqsopen}{\stackrel{\Delta}{=}}

\newcommand{\eqc}{\stackrel{\Sc}{=}}

\newcommand{\cd}[1]{\tensor{\nabla}{#1}}
\newcommand{\cds}[1]{\tensor{\overline{\nabla}}{#1}}

\newcommand{\commute}[2]{\left[#1,#2\right]}
\newcommand{\dpart}[1]{\partial_{#1}}

\newcommand{\spacef}{\ }

\newcommand{\C}{\mathcal{C}}

\newcommand{\Ah}{\hat{A}}
\newcommand{\Bh}{\hat{B}}

\newcommand{\Aa}{\acute{A}}
\newcommand{\Ba}{\acute{B}}

\newcommand{\T}{\mathcal{T}}
\newcommand{\D}{\mathcal{D}}
\newcommand{\W}{\mathcal{W}}

\newcommand{\Sc}{\mathcal{S}}
\newcommand{\Pc}{\mathcal{P}}
\renewcommand{\P}{\mathcal{P}}
\newcommand{\I}{\mathcal{I}}
\newcommand{\U}{\mathcal{U}}

\newcommand{\Cn}{\mathcal{C}}

\newcommand{\Scn}{\mathbf{S}}
\newcommand{\ms}[1]{\ct{h}{#1}}

\newcommand{\msg}[1]{\ct{\tilde{h}}{#1}}
\newcommand{\mc}[1]{\ct{q}{#1}}

\newcommand{\mcn}[1]{\ct{\ubar{q}}{#1}}
\newcommand{\mcng}[1]{\ct{\ubar{{\tilde{q}}}}{#1}}
\newcommand{\ma}{\acute{m}}
\newcommand{\mh}{\hat{m}}
\newcolumntype{M}[1]{>{\centering\arraybackslash}m{#1}}
\newcolumntype{N}{@{}m{0pt}@{}}

\newcounter{marginnotecount}[section]

	\title{\Large \textbf{Degrees of freedom of gravitational radiation}\\ \textbf{with positive cosmological constant}}
	\author[]{Francisco Fernández-Álvarez\thanks{francisco.fernandez@ehu.eus}\ }
	\affil[]{Departamento de Física\\ Universidad del País Vasco UPV/EHU\\ Apartado 644, 48080 Bilbao, Spain\\
	\ \\
	EHU Quantum Center, Universidad del País Vasco UPV/EHU}
	\date{\today{}}
\begin{document}
	
	\maketitle

	\begin{abstract}
	Results on the isolation of the radiative degrees of freedom of the gravitational field with a positive cosmological constant in full General Relativity are put forward.	Methods employed in a recent geometric characterisation of gravitational radiation are used and, inspired by Ashtekar's work on asymptotically flat space-times, a space of connections is defined. Ground differences emerge due to the space-like character of the conformal boundary, and one has to put into play a fundamental result by Friedrich concerning the initial value problem for space-times with a positive cosmological constant. Based on this, half of the radiative degrees of freedom are identified; remarkably, they utterly determine the gravitational radiation content for space-times with algebraically special rescaled Weyl tensor at infinity. Directions for defining the phase space in the general case are proposed. 
	\end{abstract}

	%\listoftheorems
%	\newpage
	\tableofcontents
%	\listoffigures
%%%%%%%%%%%%%%%%%%%%%%%%%%%%%%%%%%%%%%%%%%%%%%%%%%%%%%%%%%%%%%%%%%%%%%%%	
%%%%%%%%%%%%%%%%%%%%%%%%%%%%%%%%%%%%%%%%%%%%%%%%%%%%%%%%%%%%%%%%%%%%%%%%
%	\newpage
	\section{Introduction}\label{sec:introduction}
		Gravitational radiation is a non-linear feature of the full theory of General Relativity. Theoretical developments in the study of this phenomenon trace back to Einstein's quadrupole formula \cite{Einstein18}, followed by great achievements  in the non-linear regime during the 50s \cite{Pirani57,Bel1958,Trautman58} and 60s \cite{Sachs1962,Bondi1962,Newman62,Winicour1966}. However, it is not until the 70s when a  robust geometrical understanding was available thanks to the work of Geroch \cite{Geroch1977}. The perspective adopted there does not depend on the choice of coordinates nor any other gauge and is built using the tools of the conformal compactification by Penrose \cite{Penrose65}. Arguably, this approach is better suited to address underlying conceptual questions. As an example, Ashtekar's way of quantising the asymptotic gravitational field \cite{Ashtekar1981b} was grounded on this geometrical formulation, upon which the radiative degrees of freedom at infinity were determined \cite{Ashtekar81} and symplectic methods were developed \cite{Ashtekar1981c}. \\
		
		It is important to point out that most of these advances are only applicable when the cosmological constant vanishes and, while detection of gravitational waves \cite{LigoVirgoPRL} urges us to study in detail this aspect of gravity, evidence of a positive cosmological constant \cite{Riess1998,Perlmutter1999} in our universe requires new progress in the theory ---this issue was already exposed by Penrose in \cite{Penrose2011}. A new geometrical description of gravitational radiation at infinity in full General Relativity, that applies to the case with non-negative cosmological constant, is now available  \cite{Fernandez-AlvarezSenovilla2022b} ---see also references therein and \cite{Senovilla2022, Szabados2019} for review on other approaches too. Some of the open questions to be addressed include: the identification of a phase space of radiative modes; an abstract formulation of the characteristic problem in conformal space-time that could be connected with the presence of gravitational waves at infinity according to the new characterisation; a general formulation of energy-momentum and angular-momentum, as well as a mass-loss formula describing the energy carried by gravitational waves that arrive at infinity. This paper focuses on the first of these issues.\\
		
		In the first part of the work, a class of differential operators associated to triads of vector fields is introduced on an abstract three-dimensional Riemannian manifold. The space of all such operators is shown to have the structure of an affine space where each point is labelled by two functions. Also, it is proved that for $ C^\infty $ metrics the Levi-Civita connection belongs to that space and is determined by a set of three families of two-dimensional connections associated to a triad of vector fields. In the second part, the conformal boundary of space-times with positive cosmological constant is considered. Based on  results by Friedrich \cite{Friedrich1986a,Friedrich1986b}  and using the new characterisation of gravitational radiation in full General Relativity, it is argued how the two `coordinates' of the space of differential operators correspond to half of the gravitational radiative degrees of freedom.  Remarkably, for space-times having an algebraically special rescaled Weyl tensor at infinity, the gravitational radiation is fully determined by these two degrees of freedom. For the general case, it is suggested how another pair has to come from the TT-tensor in the conformal initial data. Some examples of exact solutions are given and, also, a comparison with Ashtekar's radiative phase space in asymptotically flat space-times is presented. 
		
	%%%%%%%%%%%%%%%%%%%%%%%%%%%%%%%%%%%%%%%%%%%%%%%%%%%%%%%%%%%%%%%%%%%%%%%%		

	\section{Two-dimensional connections in three-dimensional Riemannian manifolds}\label{sec:geometric-results}
	
	In the next two sections, an abstract 3-dimensional Riemannian manifold $ \I $ with metric $ \ms{_{ab}} $ will be considered. Let $ \C $ denote a congruence of curves on $ \I $ given by a unit vector field $ \ct{m}{^{a}} $ and parameter $ v $ that can be viewed as a function of the coordinates ---a specific coordinate system will not be introduced though. The quotient $ \Scn\defeq\I/\Cn $ is a two-dimensional differentiable\footnote{It can be the case that $ \Scn $ is not globally differentiable.} manifold, not Riemannian in general because it is not endowed with a standard metric. This $ \Scn $ is called the `projected  surface' ---the reader is referred to appendix A in \cite{Fernandez-AlvarezSenovilla2022b} for a detailed formulation. Let $ \cbrkt{\ctcn{E}{^a_{A}}} $ and $ \cbrkt{\ctcn{W}{_{a}^A}} $ be a couple of linearly independent vector fields and forms on $ \I $, orthogonal to $ \ct{m}{_{a}} $ and $ \ct{m}{^a} $ and, hence constituting a pair of dual bases on $ \Scn $, where indices $ A, B, C...=2,3 $. Using these fields, the projector to the space orthogonal to $ \ct{m}{_{a}} $ is constructed as\footnote{Underlining of tensors indicates that they are orthogonal to some vector field giving a congruence; in this case, $ \ct{m}{^a} $.}
		\begin{equation}\label{eq:projector-foliation}
			\ctcn{P}{^a_b}\defeq \ctcn{E}{^a_C}\ctcn{W}{_b^C},\quad\ctcn{P}{^c_b}\ct{m}{_c}=0=\ctcn{P}{^a_c}\ct{m}{^c},\quad\ctcn{P}{^c_c}=2\ ,
		\end{equation}
	which in terms of $ \ct{m}{_a} $ reads
		\begin{equation}\label{eq:projector-congruences}
			\ctcn{P}{_{ab}}=\ms{_{ab}}-\ct{m}{_a}\ct{m}{_b}\ .
		\end{equation}
	The acceleration, expansion tensor, and vorticity of $ \ct{m}{^a} $ are denoted by
		\begin{equation}\label{eq:kinematic}
		\ctcn{a}{_b}\defeq \ct{m}{^c}\cds{_c}\ct{m}{_b}\ ,\quad
		\ctcn{\kappa}{_{ab}}\defeq \ctcn{P}{^c_a}\ctcn{P}{^d_b}\cds{_{(c}}\ct{m}{_{d)}}\ ,\quad
		\ctcn{\omega}{_{ab}}\defeq \ctcn{P}{^c_a}\ctcn{P}{^d_b}\cds{_{[c}}\ct{m}{_{d]}}\ ,
		\end{equation}
	respectively, and the shear is defined as
		\begin{equation}\label{eq:kappa-congruence}
		\ctcn{\Sigma}{_{ab}}\defeq \ctcn{\kappa}{_{ab}}-\frac{1}{2}\ctcn{P}{_{ab}}\cscn{\kappa}\ ,\quad\cscn{\kappa}\defeq \ctcn{P}{^{cd}}\ctcn{\kappa}{_{cd}}\ .
		\end{equation}
	Since these tensors are orthogonal to $ \ct{m}{^a} $, they can be expressed as
			\begin{equation}\label{eq:kinematic-AB}
				\ctcn{a}{_a}=\ctcn{W}{_a^A}\ctcn{a}{_A}\ ,\quad
				\ctcn{\kappa}{_{ab}}=\ctcn{W}{_a^A}\ctcn{W}{_b^B}\ctcn{\kappa}{_{AB}}\ ,\quad
				\ctcn{\Sigma}{_{ab}}=\ctcn{W}{_a^A}\ctcn{W}{_b^B}\ctcn{\Sigma}{_{AB}}\ ,\quad
				\ctcn{\omega}{_{ab}}=\ctcn{W}{_a^A}\ctcn{W}{_b^B}\ctcn{\omega}{_{AB}}\ .
			\end{equation}
	It is possible to define a one-parameter family ---depending on $ v $--- of metrics $ \mcn{_{AB}} $, connections $ \ctcn{\gamma}{^A_{BC}} $ and volume forms $ \ctcn{\epsilon}{_{AB}} $ in $ \Scn $ --see appendix A in \cite{Fernandez-AlvarezSenovilla2022b}. From now on, for the sake of briefness, they are referred to as metric, connection and volume form. Nevertheless, \emph{one has to bear in mind their true significance}; it will be explicitly remarked when necessary. In particular, the metric $ \mcn{_{AB}} $ on $ \Scn $ and its inverse $ \mcn{^{AB}} $ can be written as	
		\begin{equation}
			\mcn{_{AB}}=\ctcn{E}{^a_A}\ctcn{E}{^b_{B}}\ms{_{ab}}\ ,\quad \mcn{^{AB}}=\ctcn{W}{_{a}^A}\ctcn{W}{_{b}^B}\ms{^{ab}}\ .
		\end{equation}
%	Since conformal classes $ \brkt{\ms{_{ab}}} $ are being considered, it is natural to introduce a conformal class of vector fields $ \ct{m}{^a} $ ($ \ctg{m}{^a}=1/\omega\ct{m}{^a} $ in order to preserve the condition $ \ct{m}{^a}\ct{m}{_{a}}=1 $), and hence a conformal class of metrics $ \brkt{\mcn{_{AB}}} $ on $ \Scn $; under transformation \eqref{eq:conformal-metric-transformation},
%		\begin{equation}
%			\mcng{_{AB}}=\omega^2\mcn{_{AB}}\ .
%		\end{equation}
	Also, one introduces a two-dimensional connection $ \ctcn{\gamma}{^C_{AB}} $ as
		\begin{equation}\label{eq:connection-S2}
			\ctcn{\gamma}{^C_{AB}}\defeq \ctcn{E}{^a_{A}}\ctcn{W}{_{c}^C}\cds{_{a}}\ctcn{E}{^c_{B}}\ ,
		\end{equation}
	where $ \cds{_{a}} $ stands for the covariant derivative associated to the three-dimensional Levi-Civita connection $ \cts{\Gamma}{^a_{bc}} $ of $ \prn{\I,\ms{_{ab}}} $. The connection $ \ctcn{\gamma}{^C_{AB}} $ is used to define a covariant derivative $ \cdcn{_{A}} $ on $ \Scn $ in the usual way
		\begin{equation}\label{eq:covariant-derivative-Scn}
			\cdcn{_A}\ct{v}{^B}\defeq \ctcn{E}{^a_A}\partial{_a}\ct{v}{^B}+\ctcn{\gamma}{^B_{AC}}\ct{v}{^C}, \text{ with } \ct{v}{^A}=\ct{v}{^a}\ctcn{W}{_a^A},\quad \ct{v}{^a}\ct{m}{_a}=0.
		\end{equation}
	This derivative operator is metric and volume-preserving and for a general tensor field $ \ct{T}{^{a_1 ...a_r}_{b_1...b_q}} $ on $ \I $ one has
		\begin{align}\label{eq:intrinsicdevS2}
		&\ctcn{W}{_{m_{1}}^{A_1}}...\ctcn{W}{_{m_r}^{A_r}}\ctcn{E}{^{n_1}_{B_q}}...\ctcn{E}{^{n_q}_{B_q}}\ctcn{E}{^r_C}\cds{_r}\ct{T}{^{m_1...m_r}_{n_1...n_q}}=\cdcn{_C}\ctcn{T}{^{A_1...A_r}_{B_1...B_q}}\nonumber\\
		&+\sum_{i=1}^{r}\ct{T}{^{A_1...A_{i-1}s A_{i+1}...A_r}_{B_1...B_q}}\ct{m}{_s}\prn{\ctcn{\kappa}{_C^{A_i}}+\ctcn{\omega}{_C^{A_i}}}+\sum_{i=1}^{q}\ct{T}{^{A_1...A_r}_{B_1...B_{i-1}s B_{i+1}...B_q}}\prn{\ctcn{\kappa}{_{CB_i}}+\ctcn{\omega}{_{CB_i}}}\ct{m}{^{s}}.
		\end{align}

		%%%%%%%%%%%%%%%%%%%%%%%%%%%%%%%%%%%%%%%%%%%%%%%%%%%%%%%%%%%%%%%%%%%%%%%%
	\section{Triplets of two-dimensional connections and triads}
		If one considers different families of curves on $ \prn{\I,\ms{_{ab}}} $, each one has its own projected surface. For the rest of the work, it is convenient to introduce the following definition involving three orthogonal congruences:
		\begin{deff}[Triad and triplet of connections]\label{def:triad-connections}
			Let $ \cbrkt{\ct{e}{^a_{i}}}= \cbrkt{\ct{m}{^{a}},\cta{m}{^{a}},\cth{m}{^{a}}} $ be a triad ---with $ i=1,2,3 $---, a set of unit vector fields  giving three orthogonal congruences on $ \prn{\I,\ms{_{ab}}} $ or on an open subset $ \Delta\subset\I $, and denote by $ \Scn $, $ \acute{\Scn} $ and $ \hat{\Scn} $ their associated projected surfaces. Then, a triplet of connections is a set $ \cbrkt{\cdcn{_{A}},\cdcna{_{\Aa}},\cdcnh{_{\Ah}}} $ where each element represents the one-parameter family of two-dimensional connections in $ \Scn $, $ \acute{\Scn} $ and $ \hat{\Scn} $, respectively.
		\end{deff}
	To prevent confusion, decorations are used in quantities and Latin indices associated with  $ \acute{\Scn} $ and $ \hat{\Scn} $, respectively. The following geometric identities follow by direct computation from definitions in \cref{eq:kinematic}:
		\begin{lemma}\label{thm:geom-id}
			Let $ \prn{\I,\ms{_{ab}}} $ be a three-dimensional Riemannian manifold. Then, the kinematic quantities of a triad $\cbrkt{\ct{e}{^a_{i}}}= \cbrkt{\ct{m}{^{a}},\cta{m}{^{a}},\cth{m}{^{a}}} $ obey the identities
				\begin{align}
					\ctcn{\Sigma}{_{\ma\mh}}&=\ctcnh{\omega}{_{m\ma}}+\ctcna{\omega}{_{m\mh}}\ ,\quad
					\ctcna{\Sigma}{_{m\mh}}=\ctcn{\omega}{_{\ma\mh}}+\ctcnh{\omega}{_{\ma m}}\ ,\quad
					\ctcnh{\Sigma}{_{m\ma}}=\ctcn{\omega}{_{\mh\ma}}+\ctcna{\omega}{_{\mh m}}\ ,\label{eq:omegas-sigmas}\\
					\ctcn{\Sigma}{_{\ma\ma}}&=-\ctcna{a}{_{m}}-\frac{1}{2}\ctcn{\kappa}{}=-\ctcn{\Sigma}{_{\mh\mh}}=\ctcnh{a}{_{m}}+\frac{1}{2}\ctcn{\kappa}{},\\
					\ctcna{\Sigma}{_{\mh\mh}}&=-\ctcnh{a}{_{\ma}}-\frac{1}{2}\ctcna{\kappa}{}=-\ctcna{\Sigma}{_{mm}}=\ctcn{a}{_{\ma}}+\frac{1}{2}\ctcna{\kappa}{},\\
					\ctcnh{\Sigma}{_{\ma\ma}}&=-\ctcna{a}{_{\mh}}-\frac{1}{2}\ctcnh{\kappa}{}=-\ctcnh{\Sigma}{_{mm}}=\ctcn{a}{_{\mh}}+\frac{1}{2}\ctcnh{\kappa}{}.
				\end{align}
		\end{lemma}
		\begin{remark}
		Observe that, due to the orthonormality of the triad $ \cbrkt{\ct{e}{^a_{i}}} $, one has
			\begin{equation}
			 \ct{m}{^{a}}=\ctcna{E}{^a_{\Aa}}\ct{m}{^\Aa}=\ctcnh{E}{^a_{\Ah}}\ct{m}{^\Ah}\ ,\quad \cta{m}{^{a}}=\ctcn{E}{^a_{A}}\cta{m}{^A}=\ctcnh{E}{^a_{\Ah}}\cta{m}{^\Ah}\ ,\quad \cth{m}{^{a}}=\ctcna{E}{^a_{\Aa}}\cth{m}{^\Aa}=\ctcn{E}{^a_{A}}\cth{m}{^A}\ .
			\end{equation}
		The subindices $ m $, $ \ma $ and $ \mh $ denote contraction with $ \cbrkt{\ct{m}{^{a}},\cta{m}{^{a}},\cth{m}{^{a}}} $ of the corresponding kinematic quantities. For instance, $ \ctcn{\Sigma}{_{\ma\mh}}=\ct{\ma}{^A}\ct{\mh}{^B}\ctcn{\Sigma}{_{AB}} $.
		\end{remark}
		Since the vorticities are antisymmetric, summation of the three relations in \cref{eq:omegas-sigmas} produces another identity that will be used later:
		\begin{corollary}
			Let $ \prn{\I,\ms{_{ab}}} $ be a three-dimensional Riemannian manifold. Then, the shears of a triad $\cbrkt{\ct{e}{^a_{i}}}= \cbrkt{\ct{m}{^{a}},\cta{m}{^{a}},\cth{m}{^{a}}} $ fulfil
			\begin{equation}\label{eq:identity-shears}
				\ctcn{\Sigma}{_{\ma\mh}}+\ctcna{\Sigma}{_{ m\mh}}+\ctcnh{\Sigma}{_{m\ma}}=0\ .
			\end{equation}
		\end{corollary}
	Next, one of the main results is presented,
		\begin{thm}[Derivative operator associated to a triad]\label{thm:operator-D}
			Let $ \I $ be a three-dimensional Riemannian manifold endowed with metric $ \ms{_{ab}} $, and consider a triad of unit vector fields $\cbrkt{\ct{e}{^a_{i}}}= \cbrkt{\ct{m}{^{a}},\cta{m}{^{a}},\cth{m}{^{a}}} $ as in \cref{def:triad-connections}. Then, associated to this triad, there exists a unique well-defined torsion-free differential operator $ \cdt{_{a}} $ whose action on tensor fields $ \ct{T}{_{b_1...b_n}^{c_1...c_q}} $ reads
				\begin{equation}\label{eq:operator-D}
				 	\cdt{_{a}}\ct{T}{_{b_1...b_n}^{c_1...c_q}}=\ct{D}{_{a}}\ct{T}{_{b_1...b_n}^{c_1...c_q}}+\sum_{i=1}^{n}\ct{t}{^r_{ab_{i}}}\ct{T}{_{b_1...r...b_n}^{c_1...c_q}}-\sum_{j=1}^{q}\ct{t}{^{c_j}_{ar}}\ct{T}{_{b_1...b_n}^{c_1...r...c_q}}\ ,
				\end{equation}
			where $ \ct{t}{^a_{bc}} $ is a tensor field that is antisymmetric on its covariant indices, vanishes under index contraction, depends on the vorticities of $ \cbrkt{\ct{e}{^a_{i}}} $ and can be expressed as
				\begin{equation}
					\ct{t}{^a_{bc}}\defeq 2\ct{m}{^a}\ctcn{\omega}{_{bc}}+2\cta{m}{^a}\ctcna{\omega}{_{bc}}+2\cth{m}{^a}\ctcnh{\omega}{_{bc}}\ ,
				\end{equation}
			and $ \ct{D}{_{a}}\ct{T}{_{b_1...b_n}^{c_1...c_q}} $ is a tensor field depending only on the triplet of two-dimensional connections of \cref{def:triad-connections} associated with  $ \cbrkt{\ct{e}{^a_{i}}} $ and acting on components of $ \ct{T}{_{b_1...b_n}^{c_1...c_q}} $. Also, this differential operator is related to the Levi-Civita  connection $ \cds{_{a}} $ of $ \prn{\I,\ms{_{ab}}} $ by
				\begin{equation}\label{eq:operator-D-levi-civita-connection}
				 	\cdt{_{a}}\ct{T}{_{b_1...b_n}^{c_1...c_q}}=\cds{_{a}}\ct{T}{_{b_1...b_n}^{c_1...c_q}}-\sum_{i=1}^{n}\ct{s}{^r_{ab_{i}}}\ct{T}{_{b_1...r...b_n}^{c_1...c_q}}+\sum_{j=1}^{q}\ct{s}{^{c_j}_{ar}}\ct{T}{_{b_1...b_n}^{c_1...r...c_q}}\ ,
				\end{equation}
			where $ \ct{s}{^a_{bc}} $ is a tensor field symmetric on its covariant indices, traceless in all of them, that can be written in terms of the shears of $  \cbrkt{\ct{e}{^a_{i}}}$ as
				\begin{equation}\label{eq:shear-tensor-triad}
				 	\ct{s}{^a_{bc}}\defeq 2\ctcn{\Sigma}{_{\mh\ma}}\ct{m}{^a}\cth{m}{_{(b}}\cta{m}{_{c)}}+2\ctcna{\Sigma}{_{\mh m}}\cta{m}{^a}\cth{m}{_{(b}}\ct{m}{_{c)}}+2\ctcnh{\Sigma}{_{m\ma}}\cth{m}{^a}\ct{m}{_{(b}}\cta{m}{_{c)}}\ . 
				\end{equation}
		\end{thm}
		\begin{remark}\label{rmk:diff-prop}
		By `well defined differential operator' it is meant that the typical properties --e.g., see \cite{Wald1984}-- are fulfilled: linearity, Leibnitz rule, commutativity with contraction of indices, coinciding with the partial derivative when acting on functions, and torsion-free.
		\end{remark}
		\begin{remark}\label{rmk:three-dim-comp}
		 The tensor field $ \ct{s}{^a_{bc}} $ cannot be written in terms of the triplet of connections $ \cbrkt{\cdcn{_{A}},\cdcna{_{\Aa}},\cdcnh{_{\Ah}}} $ because $ \ctcn{\Sigma}{_{\mh\ma}} $, $ \ctcna{\Sigma}{_{\mh m}} $, $ \ctcnh{\Sigma}{_{m\ma}} $ depend on `purely three-dimensional' data of the Levi-Civita connection. More concretely, the Christoffel symbols given by $ \cts{\Gamma}{^m_{\mh\ma}} $, $ \cts{\Gamma}{^\mh_{m\ma}} $ and $ \cts{\Gamma}{^\ma_{\mh m}} $ are involved.
		\end{remark}
		\begin{proof}
			Begin by introducing an auxiliary differential operator $  \ct{D}{_{a}} $ defined by
				\begin{equation}
					\ct{D}{_{a}}\ct{T}{_{b_1...b_n}^{c_1...c_q}}\defeq\cds{_{a}}\ct{T}{_{b_1...b_n}^{c_1...c_q}}-\sum_{i=1}^{n}\prn{\ct{s}{^r_{ab_{i}}}+\ct{t}{^r_{ab_{i}}}}\ct{T}{_{b_1...r...b_n}^{c_1...c_q}}+\sum_{j=1}^{q}\prn{\ct{s}{^{c_j}_{ar}}+\ct{t}{^{c_j}_{ar}}}\ct{T}{_{b_1...b_n}^{c_1...r...c_q}}\ .
				\end{equation}
			Using the properties of $ \cds{_{a}} $ it is straightforward to check that $ \ct{D}{_{a}} $ is unique and satisfies all properties of \cref{rmk:diff-prop} except for having torsion --i.e., for a general differentiable function $ f $, $ \ct{D}{_{a}}\ct{D}{_{b}}f-\ct{D}{_{b}}\ct{D}{_{a}}f=-2\ct{t}{^c_{ab}}\dpart{c}f\neq 0$. The next step is to show that $ \ct{D}{_{a}}\ct{T}{_{b_1...b_n}^{c_1...c_q}} $ depends only on the action of the triplet of two-dimensional connections of \cref{def:triad-connections} on $ \ct{T}{_{b_1...b_n}^{c_1...c_q}} $. To see this, consider first the action on a one-form $ \ct{v}{_{a}} $,
			\begin{align}
				\ct{D}{_{a}}\ct{v}{_{b}}&=\cds{_{a}}\ct{v}{_{b}}-\prn{\ct{s}{^r_{ab}}+\ct{t}{^r_{ab}}}\ct{v}{_{r}}\nonumber\\
				&=\ctcn{W}{_{a}^A}\ctcn{W}{_{b}^B}\cdcn{_{A}}\ctcn{v}{_{B}}+\cta{m}{_{a}}\cta{m}{_{b}}\cta{m}{^\Ah}\cta{m}{^\Bh}\cdcnh{_{\Ah}}\prn{\ct{m}{^e}\ct{v}{_{e}}\ct{m}{_{\Bh}}}+\cth{m}{_{a}}\cth{m}{_{b}}\cth{m}{^\Aa}\cth{m}{^\Ba}\cdcna{_{\Aa}}\prn{\ct{m}{^e}\ct{v}{_{e}}\ct{m}{_{\Ba}}}\nonumber\\
				&+ \ct{m}{_{a}}\cta{m}{_{b}}\ct{m}{^\Ah}\cta{m}{^\Bh}\cdcnh{_{\Ah}}\ctcn{v}{_{\Bh}}+\ct{m}{_{b}}\cta{m}{_{a}}\ct{m}{^\Bh}\cta{m}{^\Ah}\cdcnh{_{\Ah}}\ctcn{v}{_{\Bh}}+\ct{m}{_{a}}\cth{m}{_{b}}\ct{m}{^\Aa}\cth{m}{^\Ba}\cdcna{_{\Aa}}\ctcn{v}{_{\Ba}}+\ct{m}{_{b}}\cth{m}{_{a}}\ct{m}{^\Ba}\cth{m}{^\Aa}\cdcna{_{\Aa}}\ctcn{v}{_{\Ba}}\nonumber\\
				+&\ct{m}{_{a}}\ct{m}{_{b}}\brkt{\ct{m}{^\Aa}\ct{m}{^\Ba}\cdcna{_{\Aa}}\ctcn{v}{_{\Ba}}+\ct{m}{^\Ah}\ct{m}{^\Bh}\cdcnh{_{\Ah}}\prn{\cta{m}{_{\Bh}}\cta{m}{^e}\ct{v}{_{e}}}}\ .
			\end{align}  
			In the computation above, one decomposes $ \cds{_{a}}\ct{v}{_{b}} $ into its tangent and orthogonal parts with respect to $ \ct{m}{^a} $ ---or equivalently w.r.t. $ \cta{m}{^a} $ or $ \cth{m}{^a} $--- and then does the same with the mixed terms, but this time with respect to the other two elements of the triad ---$ \cta{m}{^a} $ and $ \cth{m}{^a} $. Afterwards, one introduces the derivative operators $ \cbrkt{\cdcn{_{A}},\cdcna{_{\Aa}},\cdcnh{_{\Ah}}} $ where possible. There are terms which cannot be written in terms of these derivatives ---for the reasons given in \cref{rmk:three-dim-comp}---, which precisely compensate the term $ \prn{\ct{s}{^r_{ab}}+\ct{t}{^r_{ab}}}\ct{v}{_{b}} $. The same type of algebraic manipulation leads to a formula for contravariant vector fields $ \ct{v}{^a} $,
			\begin{align}
				\ct{D}{_{a}}\ct{v}{^{b}}&=\cds{_{a}}\ct{v}{^{b}}+\prn{\ct{s}{^b_{ar}}+\ct{t}{^b_{ar}}}\ct{v}{^{r}}\nonumber\\
				&=\ctcn{W}{_{a}^A}\ctcn{E}{^{b}_B}\cdcn{_{A}}\ctcn{v}{^{B}}+\cta{m}{_{a}}\cta{m}{^{b}}\cta{m}{^\Ah}\cta{m}{_\Bh}\cdcnh{_{\Ah}}\prn{\ct{m}{^e}\ct{v}{_{e}}\ct{m}{^{\Bh}}}+\cth{m}{_{a}}\cth{m}{^{b}}\cth{m}{^\Aa}\cth{m}{_\Ba}\cdcna{_{\Aa}}\prn{\ct{m}{^e}\ct{v}{_{e}}\ct{m}{^{\Ba}}}\nonumber\\
				&+ \ct{m}{_{a}}\cta{m}{^{b}}\ct{m}{^\Ah}\cta{m}{_\Bh}\cdcnh{_{\Ah}}\ctcn{v}{^{\Bh}}+\ct{m}{^{b}}\cta{m}{_{a}}\ct{m}{_\Bh}\cta{m}{^\Ah}\cdcnh{_{\Ah}}\ctcn{v}{^{\Bh}}+\ct{m}{_{a}}\cth{m}{^{b}}\ct{m}{^\Aa}\cth{m}{_\Ba}\cdcna{_{\Aa}}\ctcn{v}{^{\Ba}}+\ct{m}{^{b}}\cth{m}{_{a}}\ct{m}{_\Ba}\cth{m}{^\Aa}\cdcna{_{\Aa}}\ctcn{v}{^{\Ba}}\nonumber\\
				&+\ct{m}{_{a}}\ct{m}{^{b}}\brkt{\ct{m}{^\Aa}\ct{m}{_\Ba}\cdcna{_{\Aa}}\ctcn{v}{^{\Ba}}+\ct{m}{^\Ah}\ct{m}{_\Bh}\cdcnh{_{\Ah}}\prn{\cta{m}{^{\Bh}}\cta{m}{^e}\ct{v}{_{e}}}}+\nonumber\\
				&+\ms{^{rb}}\prn{\ct{s}{^e_{ar}}+\ct{t}{^e_{ar}}}\ct{v}{_{e}}+\prn{\ct{s}{^b_{ar}}+\ct{t}{^b_{ar}}}\ct{v}{^{r}}\ .
			\end{align}  	
			By application of \cref{eq:omegas-sigmas}, the last line vanishes. One proceeds in the same way with tensor fields of arbitrary rank.  This shows that $ \ct{D}{_{a}}\ct{T}{_{b_1...b_n}^{c_1...c_q}} $ depends only on the action of $ \cbrkt{\cdcn{_{A}},\cdcna{_{\Aa}},\cdcnh{_{\Ah}}} $ on $ \ct{T}{_{b_1...b_n}^{c_1...c_q}} $. Then, using $ \ct{D}{_{a}}\ct{T}{_{b_1...b_n}^{c_1...c_q}} $ in \cref{eq:operator-D}, it follows that $ \cdt{_{a}}\ct{T}{_{b_1...b_n}^{c_1...c_q}} $ is the desired torsion-free differential operator. 
		\end{proof}
	There is a classical result ---see\footnote{See also Corollary 3.4 in \cite{Bryant1991} where analicity of the metric is required, or \cite{Kowalski2013} for an extension to pseudo-Riemannian manifolds.} theorem 4.2 in \cite{DeTurck1984}---  that plays an important role in what comes next, 
		\begin{thm}[Classical metric diagonalisation result.]\label{thm:classic}
			Let $ \I $ be a three-dimensional Riemannian manifold endowed with metric $ \ms{_{ab}} $ of class $ C^\infty $. Then, every point $ p $ of $ \I $ lies in a neighbourhood $ \U_{p} $ on which there exists a $ C^\infty $ coordinate chart $ \cbrkt{x^{\ubar{a}}} $ in which the metric takes the diagonal form
				\begin{equation}
					h=\ct{h}{_{\ubar{1}\ubar{1}}}\df{x}^{\ubar{1}}\df{x}^{\ubar{1}}+\ct{h}{_{\ubar{2}\ubar{2}}}\df{x}^{\ubar{2}}\df{x}^{\ubar{2}}+\ct{h}{_{\ubar{3}\ubar{3}}}\df{x}^{\ubar{3}}\df{x}^{\ubar{3}}\ .
				\end{equation}
		\end{thm}
	This, together with \cref{thm:operator-D}, implies the   following result:
		\begin{corollary}\label{thm:corollary-levi-civita}
			Let $ \I $ be a three-dimensional Riemannian manifold endowed with metric $ \ms{_{ab}} $ of class $ C^\infty $. Then, every point $ p $ of $ \I $ lies in a neighbourhood $ \U_{p} $ on which there exists a triad of unit vector fields $ \cbrkt{\ct{e}{^a_{i}}}=\cbrkt{\ct{m}{^a},\cta{m}{^a},\cth{m}{^a}} $ such that the associated differential operator $ \cts{\D}{_{a}} $ of \cref{thm:geom-id} coincides with the Levi-Civita connection of $ \prn{\I,\ms{_{ab}}} $:
				\begin{equation*}
					\cdt{_{a}}\ct{T}{_{b_1...b_n}^{c_1...c_q}}\equp\cds{_{a}}\ct{T}{_{b_1...b_n}^{c_1...c_q}}\ ,
				\end{equation*}
			and its action on arbitrary tensor fields $ \ct{T}{_{b_1...b_n}^{c_1...c_q}} $ depends only on the triplet of two-dimensional connections $ \cbrkt{\cdcn{_{A}},\cdcna{_{\Aa}},\cdcnh{_{\Ah}}} $ of \cref{def:triad-connections} associated to that triad of vector fields.
		\end{corollary}
		\begin{proof}
			Using the coordinate chart $ \cbrkt{x^{\ubar{a}}} $ of \cref{thm:classic}, one has that the triad given by
				\begin{equation}
					\ct{m}{^{\ubar{a}}}=\frac{1}{\sqrt{\ms{_{\ubar{1} \ubar{1}}}}}\delta^{\ubar{a}}_{\ubar{1}}\ ,\quad
					\cta{m}{^{\ubar{a}}}=\frac{1}{\sqrt{\ms{_{\ubar{2} \ubar{2}}}}}\delta^{\ubar{a}}_{\ubar{2}}\ ,\quad
					\cth{m}{^{\ubar{a}}}=\frac{1}{\sqrt{\ms{_{\ubar{3} \ubar{3}}}}}\delta^{\ubar{a}}_{\ubar{3}}
				\end{equation}
			has vanishing vorticities --see also \cref{rmk:directions}--
				\begin{equation}
				\ctcn{\omega}{_{ab}}=\ctcna{\omega}{_{ab}}=\ctcnh{\omega}{_{ab}}=0\ .
				\end{equation}
			Then, $ \ct{t}{^a_{bc}}=0 $. According to $ \cref{thm:operator-D} $, this proves that the action of the associated $ \cts{\D}{_{a}} $ on an arbitrary tensor field only depends on the triplet of two dimensional connections. Also, a direct calculation shows that
				\begin{equation}
					 \cts{\Gamma}{^m_{\mh\ma}}\equp \cts{\Gamma}{^\mh_{m\ma}}\equp  \cts{\Gamma}{^\ma_{\mh m}}\equp 0\ .
				\end{equation}
			Using this, one has
				\begin{equation}
					\ctcn{\Sigma}{_{\ma\mh}}=\ctcna{\Sigma}{_{\mh m}}=\ctcnh{\Sigma}{_{\ma m}}=0\ ,
				\end{equation}
			which, alternatively, could have been proved by using $ \ct{t}{^a_{bc}}=0 $ in \cref{eq:omegas-sigmas}.
			Hence, $ \ct{s}{^a_{bc}}=0 $ and, as follows from \cref{thm:operator-D},
				\begin{equation}
					\cdt{_{a}}\ct{T}{_{b_1...b_n}^{c_1...c_q}}\equp\cds{_{a}}\ct{T}{_{b_1...b_n}^{c_1...c_q}}\ .
				\end{equation}
		\end{proof}
		\begin{remark}\label{rmk:levi-civitta}
			By this result, the action of the covariant derivative $ \cds{_{a}} $ defined by the Levi-Civita connection of $ \prn{\I,\ms{_{ab}}} $ is determined in every neighbourhood $ \U_{p} $ by the action of a triplet of two-dimensional connections. In principle, working by patches, one can cover $ \I $, however different charts would have to be used, meaning that the metric would be diagonalised by different triads  on different regions of $ \I $. Apart from that, it is interesting to notice that all the curvature-related quantities of $ \prn{\I,\ms{_{ab}}} $ can then be computed on each $ \U_{p} $ by the action of three two-dimensional connections.
		\end{remark}
		\begin{remark}\label{rmk:directions}
			An immediate consequence of  \cref{thm:classic} is that a triad  $\cbrkt{\ct{e}{^a_{i}}}= \cbrkt{\ct{m}{^a},\cta{m}{^a},\cth{m}{^a}} $ associated to the $ \cts{\D}{_{a}} $ that coincides with $ \cds{_{a}} $ on each $ \U_{p} $ according to \cref{thm:operator-D} reads
				\begin{align}
					\ct{m}{_{\ubar{a}}}&=\sqrt{\ms{_{\ubar{1} \ubar{1}}}}\dpart{\ubar{a}}x^{\ubar{1}}\ ,\\
					\cta{m}{_{\ubar{a}}}&=\sqrt{\ms{_{\ubar{2} \ubar{2}}}}\dpart{\ubar{a}}x^{\ubar{2}}\ ,\\
					\cth{m}{_{\ubar{a}}}&=\sqrt{\ms{_{\ubar{3} \ubar{3}}}}\dpart{\ubar{a}}x^{\ubar{3}}\ .
				\end{align}
			These forms, each one being proportional to a gradient,  have vanishing vorticities there. Thus, they characterise a set of 3 orthogonal \textit{foliations} on $ \U_{p} $. \\
		\end{remark}
		A natural thing to do is to introduce a connection associated to $ \cdt{_{a}} $ as
			\begin{equation}
				\cts{\gamma}{^a_{bc}}\defeq\cts{\Gamma}{^a_{bc}}+\ct{s}{^a_{bc}}\ ,
			\end{equation}
		such that
			\begin{equation}
				\cdt{_{a}}\ct{v}{^{b}}=\dpart{a}\ct{v}{^{b}}+\cts{\gamma}{^b_{ac}}\ct{v}{^{c}}.
			\end{equation}
		Some properties of $ \cdt{_{a}} $ that can be computed straightforwardly are
			\begin{align}
				\cdt{_{[a}}\cdt{_{b]}}f&=0\ ,\\
				\cdt{_{[a}}\ct{v}{_{b]}}&=\cds{_{[a}}\ct{v}{_{b]}}=\dpart{[a}\ct{v}{_{b]}}\ ,\\
				\cds{_{[a}}\cds{_{b]}}\ct{v}{_{c}}&=\cdt{_{[a}}\cdt{_{b]}}\ct{v}{_{c}}+ \cdt{_{[a}}\prn{\ct{s}{^r_{b]c}}}\ct{v}{_{r}}+\ct{s}{^f_{c[a}}\ct{s}{^r_{b]f}}\ct{v}{_{r}}\ \nonumber\\
				&=\cdt{_{[a}}\cdt{_{b]}}\ct{v}{_{c}}+ \cds{_{[a}}\prn{\ct{s}{^r_{b]c}}}\ct{v}{_{r}}+\ct{s}{^f_{c[b}}\ct{s}{^r_{a]f}}\ct{v}{_{r}}\ ,\\
				\cdt{_{a}}\ms{_{bc}}&=\ct{s}{^r_{bc}}\ms{_{ra}}\ .	
			\end{align}
		Let $ \cbrkt{\ct{e}{^a_{i}}} $ be a triad defining $ \cdt{_{a}} $. Then, for any $ \ct{v}{_{a}}=\ctcn{v}{_{A}}\ctcn{W}{^A_{a}} $, 
			\begin{equation}\label{eq:nabla-D-D}
				\ctcn{E}{^a_{A}}\ctcn{E}{^b_{B}}\cds{_{a}}\ct{v}{_{b}}=\ctcn{E}{^a_{A}}\ctcn{E}{^b_{B}}\cdt{_{a}}\ct{v}{_{b}}=\cdcn{_{A}}\ctcn{v}{_{B}}\ ,
			\end{equation}
		and the decorated version of this relation holds for $ \ct{v}{_{a}}=\ctcn{v}{_{\Ah}}\ctcnh{W}{^\Ah_{a}} $ and $ \ct{v}{_{a}}=\ctcn{v}{_{\Aa}}\ctcna{W}{^\Aa_{a}} $, respectively.\\
		
		To conclude this section, consider conformal transformations of the metric $ \ms{_{ab}} $,
						\begin{equation}\label{eq:gauge-ms}
							\msg{_{ab}}=\omega^2\ms{_{ab}}\ ,
						\end{equation}
					where $ \omega $ is a positive definite function. This change also affects the metric on $ \Scn $,
						\begin{equation}\label{eq:gauge-mcn}
							\mcng{_{AB}}=\omega^2\mcn{_{AB}}\ .
						\end{equation}
					Accordingly, the kinematic quantities \cref{eq:kinematic} transform too ---see Appendix E in \cite{Fernandez-AlvarezSenovilla2022b}. In particular, one has that $ \ctg{m}{^a}=\omega^{-1}\ct{m}{^a} $ and
						\begin{equation}
							\ctcng{\Sigma}{_{AB}}=\omega \ctcn{\Sigma}{_{AB}}\ ,\quad \ctcng{\omega}{_{AB}}=\omega\ctcn{\omega}{_{AB}}\ .
						\end{equation}
					This should be expected, since umbilicity and surface-orthogonality are conformally invariant properties. Thus, the conformal invariance of $ \ct{s}{^a_{bc}} $ and $ \ct{t}{^a_{bc}} $ follows,
						\begin{equation}\label{eq:gauge-s-t}
							\ctg{s}{^a_{bc}}=\ct{s}{^a_{bc}}\ ,\quad \ctg{t}{^a_{bc}}=\ct{t}{^a_{bc}}\ ,
						\end{equation}
					 and, also, this shows that  $ \cts{\gamma}{^a_{bc}} $ undergoes the same change as $ \cts{\Gamma}{^a_{bc}} $:
						\begin{equation}\label{eq:gauge-gamma-triad}
							\ctsg{\gamma}{^a_{bc}} \eqs \cts{\gamma}{^a_{bc}} +\cts{C}{^a_{bc}}\ ,\quad\cts{C}{^a_{bc}}=\frac{1}{\omega}\ms{^{at}}\prn{2\ms{_{
												t(b}}\cts{\omega}{_{c)}}-\ms{_{cb}}\cts{\omega}{_t}}\ ,
						\end{equation}
					where $ \cts{\omega}{_{a}}\defeq\cds{_{a}}\omega $. All the results above, including \cref{thm:geom-id,thm:operator-D,thm:corollary-levi-civita} are invariant under these kind of transformations.
		\subsection{The space of  connections}
			Different triads can lead to different operators $ \cdt{_{a}} $, according to \cref{thm:operator-D}. To study this, start with the following definition:
			\begin{deff}[Space of connections]
				Define the space $ \Xi $   as the set of all possible differential operators $ \cdt{_{a}} $ of \cref{thm:operator-D} on $ \I $ or on an open $ \Delta\in\I $\ .
			\end{deff}
		 From now on, no distinction between $ \I $ and $ \Delta\in\I $ will be made, but one has to take into account that these operators  may not be defined globally according to \cref{def:triad-connections}. The difference between any two elements of $ \Xi $ can be characterised as the difference of their action on an arbitrary form $ \ct{v}{_{b}} $ on $ \I $. From \cref{thm:operator-D} one has
				\begin{equation}
					\prn{\cdtp{_{a}}-\cdt{_{a}}}\ct{v}{_{b}}=\ct{d}{^e_{ab}}\ct{v}{_{e}}\ ,\quad\ct{d}{^e_{ab}}\defeq\ct{s}{^e_{ab}}-\ct{s'}{^e_{ab}}\ .
				\end{equation}
		To know how many functions have to be specified to determine a point in $ \Xi $, one just has to study the algebraic properties of $ \ct{d}{^a_{bc}} $. First note that lowering the contravariant index  $ \ct{s}{_{abc}}=\ms{_{da}}\ct{s}{^d_{bc}} $ one has
				\begin{properties}[label=\roman*)]
				\item $ \ct{s}{_{abc}}=\ct{s}{_{acb}} $\ . \label{it:propi} 
				\item $ \ct{s}{_{abc}}+\ct{s}{_{bca}}+\ct{s}{_{cab}}=0 $\ . \label{it:propii}
				\item $ \ms{^{ab}}\ct{s}{_{abc}}=0 $\ .\label{it:propiii}
				\item $ \ms{^{ac}}\ct{s}{_{abc}}=0 $\ .\label{it:propiv}
				\item $ \ct{s}{_{[abc]}}=0=\ct{s}{_{(abc)}} $\ .		\label{it:propv}
				\item $ \ct{P}{^{ab}}\ct{s}{_{abc}}=0=\ct{P}{^{ac}}\ct{s}{_{abc}} $, for any of the three projectors associated to the elements of the triad $\cbrkt{\ct{e}{^a_{i}}}= \cbrkt{\ct{m}{^{a}},\cta{m}{^{a}},\cth{m}{^{a}}} $ defining $ \ct{s}{_{abc}} $. \label{it:propvi}
				\end{properties}
			Not all the properties are independent from each other: from \cref{it:propi,it:propii,it:propiii}, one can derive \cref{it:propiv,it:propv}; whereas \ref{it:propvi} is independent from the others. \Cref{it:propi,it:propii,it:propiii} provide with 9, 10 and 3 constraints, respectively. And from \ref{it:propvi} one gets 3 additional independent equations, giving a total of 25 constraints that leave 2 independent components in $ \ct{s}{_{abc}} $. The same can be shown if one looks to \cref{eq:shear-tensor-triad}; it turns out that all the information that has to be specified is encoded in 3 scalars: $ \ctcn{\Sigma}{_{\ma\mh}}$, $\ctcna{\Sigma}{_{\mh m}}$ and $\ctcnh{\Sigma}{_{\ma m}} $. But recall that from \cref{thm:geom-id} one has \cref{eq:identity-shears}, which reduces the number of independent components to just 2 functions. However, this counting does not apply directly to $ \ct{d}{_{abc}}=\ms{_{da}}\ct{d}{^{d}_{bc}} $. The reason is that while \cref{it:propi,it:propii,it:propiii,it:propiv,it:propv} are satisfied by this tensor too, \cref{it:propvi} does not. Still, it is possible to find three additional independent constraints. One proceeds first by writing $ \ct{d}{_{abc}} $ using a generic basis $ \cbrkt{\ct{\tilde{e}}{^a_{i}}} $ --with $ i $ running from 1 to 3--,
				\begin{align}
					\ct{d}{_{ijk}}=&2\ct{\Sigma}{}\ct{\lambda}{^3_{i}}\ct{\lambda}{^1_{(j}}\ct{\lambda}{^2_{k)}}+2\cta{\Sigma}{}\ct{\lambda}{^1_{i}}\ct{\lambda}{^2_{(j}}\ct{\lambda}{^3_{k)}}+2\cth{\Sigma}{}\ct{\lambda}{^2_{i}}\ct{\lambda}{^1_{(j}}\ct{\lambda}{^3_{k)}}\nonumber\\
					-&2\ct{\Sigma'}{}\ct{\omega}{^3_{i}}\ct{\omega}{^1_{(j}}\ct{\omega}{^2_{k)}}-2\ctap{\Sigma}{}\ct{\omega}{^1_{i}}\ct{\omega}{^2_{(j}}\ct{\omega}{^3_{k)}}-2\cthp{\Sigma}{}\ct{\omega}{^2_{i}}\ct{\omega}{^1_{(j}}\ct{\omega}{^3_{k)}}\ ,
				\end{align}
			where $ \ct{\lambda}{^i_{j}} $ and $ \ct{\omega}{^i_{j}} $ relate the basis $ \cbrkt{\ct{\tilde{e}}{^a_{i}}} $ with $\cbrkt{ \ct{e}{^a_{i}}}=\cbrkt{\cta{m}{^{a}},\cth{m}{^a},\ct{m}{^a}} $ and $\cbrkt{ \ct{e'}{^a_{i}}}= \cbrkt{\ctap{m}{^{a}},\cthp{m}{^a},\ct{m'}{^a}}  $, respectively:
				\begin{align}
					\ct{\tilde{e}}{^a_{i}}&=\ct{\lambda}{^j_i}\ct{e}{^a_{j}}\ ,\\
					\ct{\tilde{e}}{^a_{i}}&=\ct{\omega}{^j_i}\ct{e'}{^a_{j}}\ .
				\end{align}
			 Each of these transformation matrices can be expressed in terms of three Euler angles\footnote{The conventions for the definition of the Euler angles  of \cite{Marion2004} are used.}; denote them by $ \cbrkt{\phi,\theta,\psi} $  and $ \cbrkt{\beta,\gamma,\alpha} $, each set defining $ \ct{\lambda}{^i_{j}} $ and $ \ct{\omega}{^i_{j}} $, respectively. Then, one can express explicitly  the components $ \ct{d}{_{ijk}} $ as functions of $ \Sigma, \cta{\Sigma}{}, \cth{\Sigma}{},\ct{\Sigma'}{}, \ctap{\Sigma}{}, \cthp{\Sigma}{}$ and the two sets of Euler angles. Observe that the following decomposition applies,
			 	\begin{equation}
			 		\ct{d}{_{ijk}}=\ct{A}{_{ijk}}\prn{\phi,\theta,\psi,\Sigma,\cta{\Sigma}{},\hat{\Sigma}}+\ct{B}{_{ijk}}\prn{\beta,\gamma,\alpha,\ct{\Sigma'}{},\ctap{\Sigma}{},\cthp{\Sigma}{}}\ .
			 	\end{equation}
			 Using this, a direct algebraic manipulation shows that\footnote{\Cref{eq:identity-shears} has to be used to derive \cref{eq:constraint3}.}
			 	\begin{align}
			 		\prn{\dpart{\psi}+\dpart{\alpha}}\ct{d}{_{133}}&=-\ct{d}{_{211}}\ ,\label{eq:constraint1}\\
			 		\prn{\dpart{\psi}+\dpart{\alpha}}\ct{d}{_{123}}&=-\ct{d}{_{322}}\ ,\label{eq:constraint2}\\
			 		\ct{d}{_{231}}&=-\ct{d}{_{123}}-2\int\ct{A}{_{322}}\df\psi-2\int\ct{B}{_{322}}\df\alpha \label{eq:constraint3} .
			 	\end{align}
				Summing up, \cref{it:propi,it:propii,it:propiii,it:propiv,it:propv} can be used to reduce the independent components to: $ \ct{d}{_{123}} $, $ \ct{d}{_{231}} $, $ \ct{d}{_{133}} $, $ \ct{d}{_{211}} $, $ \ct{d}{_{322}} $. Now, using \cref{eq:constraint1,eq:constraint2,eq:constraint3}, it is enough to specify 2 components ---$ \ct{d}{_{133}} $ and $ \ct{d}{_{123}} $--- to determine the other 3. These properties can be straightforwardly generalised to any linear combination of tensor fields  $ \ct{d}{_{abc}} $ with constant coefficients.
			  \\	
			  		
			  It is interesting to see that $ \Xi $ has the structure of an affine space\footnote{See \cite{Berger1987} for example for the definition and nomenclature used here.} with underlying vector space $ \vec{\Xi} $ (the vector space of tensors $ \ct{d}{^a_{bc}} $) over $ \mathbb{R} $. One can define the map $ \Phi $ of `translations' of $ \Xi $ as --omitting the indices in parenthesis containing $ \ct{d}{^a_{bc}} $ and $ \cdt{_{a}} $, for clearness--
		  		\begin{equation}
			  		\Phi(\overline{\D},d)_{a}\ct{v}{_{b}}=\cdt{_{a}}\ct{v}{_{b}}+\ct{d}{^e_{ab}}\ct{v}{_{e}}\ . 
		  		\end{equation}
		  	And for all $ \cdt{_{a}} $ and $ \cdtp{_{a}} $ in $ \Xi $ there exists another map $ d $: $ \Xi\times\Xi\rightarrow\vec{\Xi} $ such that $ \cdtp{_{a}}\ct{v}{_{b}}=\cdt{_{a}}\ct{v}{_{b}}+\ct{d(\cdtp{},\cdt{})}{^e_{ab}}\ct{v}{_{e}} $. Observe that the symbol $ \ct{d(\cdtp{},\cdt{})}{^e_{ab}} $ is used in an abuse of notation, since $ \ct{d}{^e_{ab}} $ are the elements of $ \vec{\Xi} $. Notice also that, for any $ \cdt{_{a}} $, $ \cdtp{_{a}} $ and $ \ct{{\overline{\D}''}}{_{a}} $ in $ \Xi $,
		  		\begin{equation}
		  			\prn{\ct{d( \cdtp{},\cdt{})}{^e_{ab}}+\ct{d(\cdt{},\ct{{\overline{\D}''}}{})}{^e_{ab}}}\ct{v}{_{e}}=\ct{d( \cdtp{},\ct{{\overline{\D}''}}{})}{^e_{ab}}\ct{v}{_{e}}\ .
		  		\end{equation}
		  		\begin{figure}[h!]
		  			\centering
			  		\includegraphics[scale=0.3]{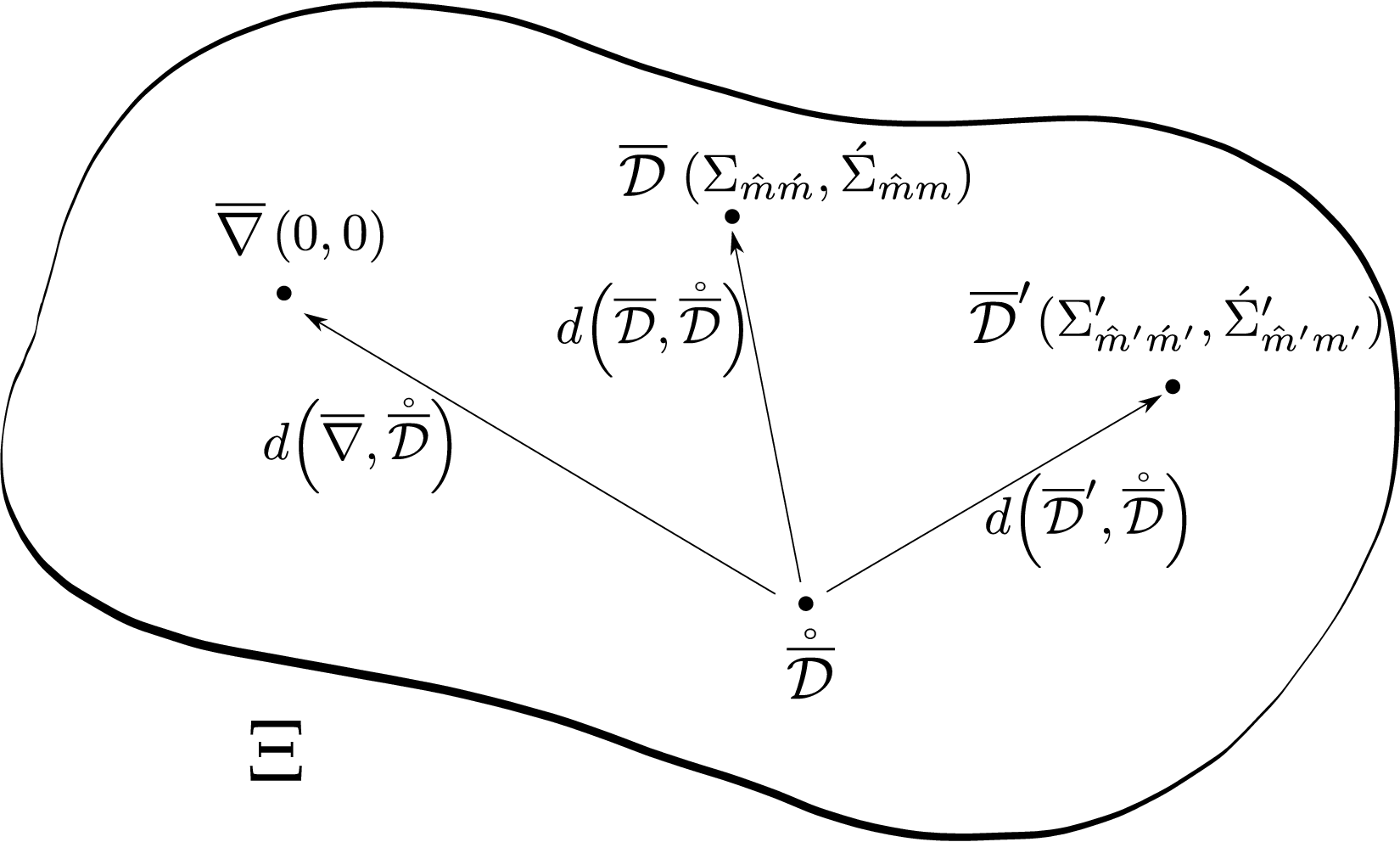}
			  		\caption{The space $ \Xi $ of differential operators is represented. Each $ \cdt{_{a}} $ is a point distancing from an arbitrary origin $ \mathring{\overline{\mathcal{D}}}_{a} $ by $ d\prn{\cdt{},\mathring{\mathcal{D}}_{}} $, the map to the vector space $ \vec{\Xi} $ of tensors $ \ct{d}{^a_{bc}} $, 	determined by two functions ---either one specifies $ \ct{d}{_{123}} $ and $ \ct{d}{_{133}} $ in an arbitrary basis, as functions of the shears and Euler angles, or considers $ \cto{s}{^a_{bc}} $ fixed and then has to specify the two components that determine $ \ct{s}{^a_{bc}} $. Independently of the choice of origin, each point can be labelled by the components of the shears $ \ct{\Sigma}{_{\ma\mh}} $ and $ \cta{\Sigma}{_{m\mh}} $, which can be viewed as `coordinates' on $ \Xi $. The point that represents the Levi-Civita connection $ \cds{_{a}} $ has coordinates $ \prn{0,0} $.}\label{fig:xi-space}
		  		\end{figure}
			Also, one can fix an origin\footnote{The choice of origin is arbitrary, so `vectorialising' $ \Xi $ can be seen as a gauge fixing from the physical point of view.} $ \cdto{_{a}} $ in $ \Xi $  ---see \cref{fig:xi-space}. Then, any other $ \cdt{_{a}} $ is characterised with respect to this origin by a tensor field $ \ct{s}{^a_{bc}} $, fulfilling \cref{it:propi,it:propii,it:propiii,it:propiv,it:propv,it:propvi} ---no matter the particular choice one does, the number of functions that coordinatize $ \Xi $ is two. One could think that the natural choice of origin in $ \Xi $ \emph{on each neighbourhood} $ \U_{p} $ is given by the kind of triad of \cref{rmk:directions} that ---at least locally--- diagonalises the metric, for which one has $ \ct{\mathring{s}}{^a_{bc}}\equp 0 $, $ \cdt{_{a}}\equp\cds{_{a}} $:
				\begin{equation}
				\prn{\cdto{_{a}}-\cdt{_{a}}}\ct{v}{_{b}}\equp\ct{s}{^e_{ab}}\ct{v}{_{e}}\ .
				\end{equation}
%			This has the inconvenient of not being the natural choice globally, in general, as $ \ct{\mathring{s}}{^a_{bc}} $ for a fixed $ \cbrkt{\ct{e}{^a_{i}}} $ might not vanish ---if it exists--- outside the given $ \U_{p}$.
				\begin{remark}\label{rmk:onto}
					Importantly, one has that a triad $ \cbrkt{\ct{e}{^a_{i}}}=\cbrkt{\ct{m}{^{a}},\cta{m}{^{a}},\cth{m}{^{a}}}  $ defines a unique $ \cts{\D}{_{a}} $, but the contrary is not true. Choosing $ \ct{\omega}{_{i}^j}=\delta_{i}^j $, one has that
						\begin{equation}
						\phi=a\frac{\pi}{2}\ ,\quad \theta=b\frac{\pi}{2}\ ,\quad\psi=c\frac{\pi}{2}\ ,
						\end{equation}
					with $ a $, $ b $, $ c $ natural numbers, leads to $ \ct{d}{_{abc}}=0 $. That, is  two triads differing by a flip of sign and/or a swap of one or more elements produce the same $ \ct{s}{^a_{bc}} $, and hence the same $ \cts{\D}{_{a}} $. Hence, one has that the map $ \varphi $ from the space of possible triads $ \Upsilon $ to the set of triplets of connections $ \Xi $
					\begin{align}
						\varphi:\Upsilon&\longrightarrow\Xi&\nonumber\\
						\cbrkt{\ct{e}{^a_{i}}}&\longrightarrow\cts{\D}{_{a}}&\nonumber
					\end{align}			
					is a surjection ---see \cref{fig:surjection}. It is possible, indeed, to have more general transformations of a given triad that do not modify the associated $ \cdt{_{a}} $; an example is having two different triads for which $ \ct{s}{^a_{bc}}=0 $. This is illustrated in the case of the C-metric in \cref{ssec:special-case}.
				\end{remark}
			\begin{figure}[h!]
							\centering
							\includegraphics[scale=0.35]{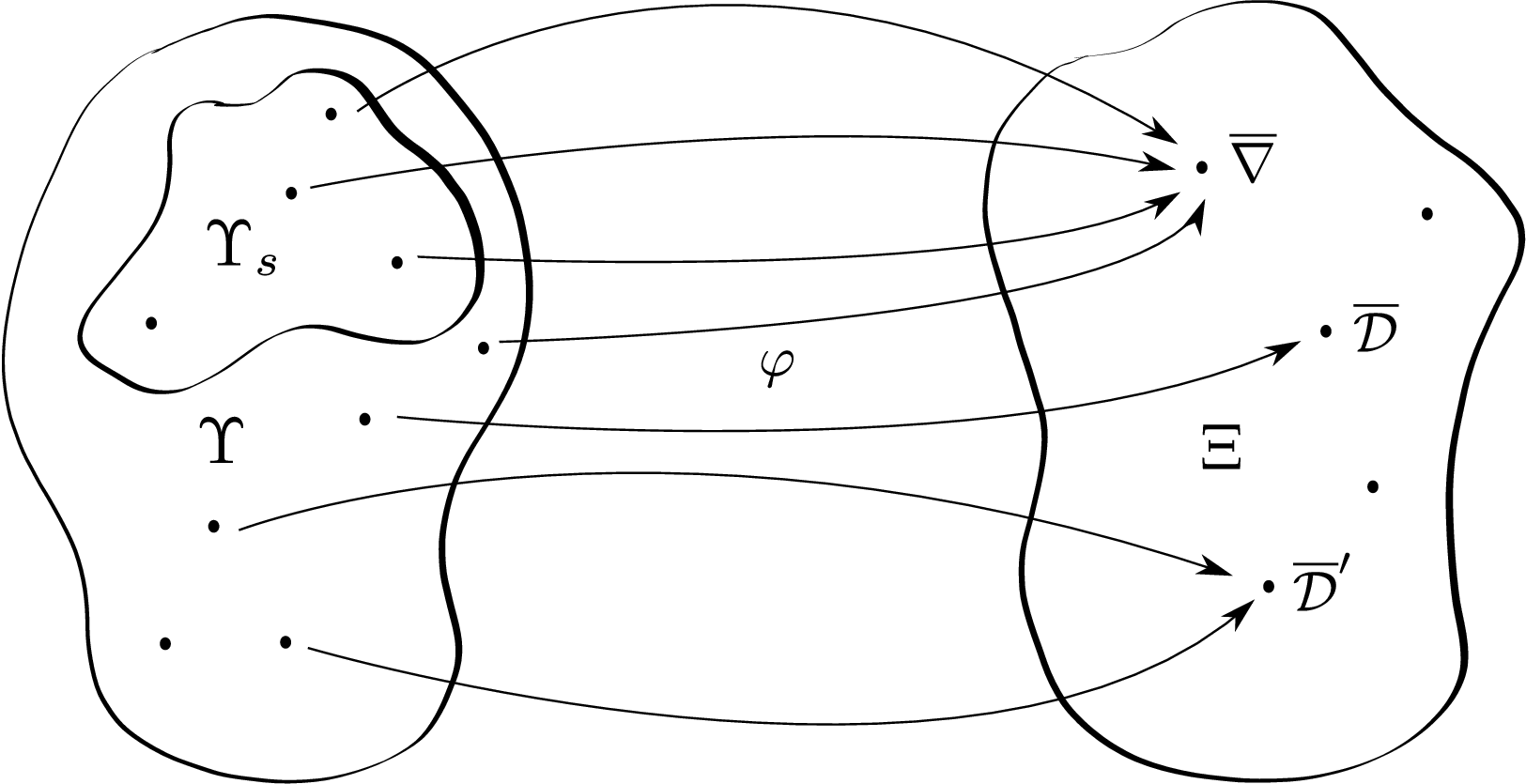}
							\caption{The surjection $ \varphi $ maps triads of vector fields in $ \Upsilon $ to differential operators in $ \Xi $. There is a subset $ \Upsilon_{s}\subset\Upsilon $ of triads whose image is $ \cds{_{a}} $ and their elements are surface-orthogonal ---i.e., they have $ \ct{t}{^a_{bc}}=0 $. In a neighbourhood of each point in $ \I $, it is always possible to find a triad that diagonalises the metric and that, according to \cref{thm:corollary-levi-civita}, belongs to $ \Upsilon_{s} $. There might be, however, triads that do not belong to $ \Upsilon_{s} $ but are mapped to $ \cds{_{a}} $ ---those having $ \ct{s}{^a_{bc}}=0 $ but $ \ct{t}{^a_{bc}}\neq 0 $.}\label{fig:surjection}
						\end{figure}	
			Introduce now the following subset of triads,
				\begin{deff}
					The strict set of triads is denoted by $ \Upsilon_{s}\subset\Upsilon $ and consists of those triads $  \cbrkt{\ct{e}{^a_{i}}} $ for which $ \ct{t}{^a_{bc}}=0 $.
				\end{deff}
			 This subset $ \Upsilon_{s} $ is of interest as it can be put in connection with the existence of the so called first components of news --see  \cref{sec:infinity} and \cref{thm:news-triads} later. Observe that by \cref{thm:geom-id} one has that
			 	\begin{equation}
			 		\ct{t}{^a_{bc}}=0\implies \ct{s}{^a_{bc}}=0
			 	\end{equation} 
			 and therefore that the restriction $ \varphi_{s} $ of $ \varphi $ acting on $ \Upsilon_{s} $ maps every triad to $ \cds{_{a}} $,
%				\begin{equation}
%					\forall \cbrkt{\ct{e}{^a_{i}}}\in \Upsilon_{s}\ ,\quad \varphi: \cbrkt{\ct{e}{^a_{i}} }\longrightarrow \cds{_{a}}\ .
%				\end{equation}
				\begin{align}
					\varphi_{s}:\Upsilon_{s}&\longrightarrow\Xi&\nonumber\\
					\cbrkt{\ct{e}{^a_{i}}}&\longrightarrow\cds{_{a}}&\nonumber
				\end{align}
			\begin{remark}\label{rmk:sset-notempty}
				For any $ C^\infty $ metric $ \ms{_{ab}} $, one has  $ \Upsilon_{s}\neq \emptyset $, since the triad given by \cref{rmk:directions}, which diagonalises the metric according to \cref{thm:classic}, has $ \ct{t}{^a_{bc}}=0 $ and, consequently, belongs to $ \Upsilon_{s}$.
			\end{remark}	
				
			 As a final note, consider the conformal rescalings \eqref{eq:gauge-ms}. While  $ \cdt{_{a}} $ changes according to \cref{eq:gauge-gamma-triad}, `distances' between points in $ \Xi $ are preserved, that is:
				\begin{equation}\label{eq:gauge-distance-xi}
					\ct{d(\tilde{\cdtp{}},\tilde{\cdt{}})}{^e_{ab}}=\ct{d(\cdtp{},\cdt{})}{^e_{ab}}\ ,
				\end{equation}
			as follows from \cref{eq:gauge-s-t}.
	\section{Radiative degrees of freedom}\label{sec:infinity}
	
		Results presented so far apply to general three-dimensional Riemannian manifolds, some of them requiring a $ C^\infty $ metric. This second part of the article is devoted to show their connection with the characterisation of gravitational radiation in full General Relativity with a positive cosmological constant. Specifically, the aim is to work out a strategy for determining the radiative degrees of freedom at infinity. To that end, consider a physical space-time\footnote{Greek letters are used for abstract indices of quantities in four dimensional space-time. The signature of the metric is chosen to be $ \prn{-,+,+,+} $.} $ \prn{\hat{M},\pt{g}{_{\alpha\beta}}} $ admitting a conformal completion à la Penrose $ \prn{M,\ct{g}{_{\alpha\beta}}} $, where the unphysical metric is related to the physical one by $ \ct{g}{_{\alpha\beta}}=\Omega^2\pt{g}{_{\alpha\beta}} $. The conformal factor $ \Omega $ is strictly positive in $ \hat{M} $, vanishes at the conformal boundary\footnote{For a detailed description see e.g. \cite{Kroon2016}.}  $ \scri $ and is not unique; there exist the freedom of rescaling $ \Omega $ by a positive definite function $ \omega $. This ambiguity is gauge, and as such, physical statements must not depend on the particular fixing\footnote{From this point of view, the fact that these transformations preserve $ \Xi $ ---in the sense of \cref{eq:gauge-distance-xi}--- is a favourable feature.}. As already anticipated in the introduction, the case of interest to this work is that of $ \pt{g}{_{\alpha\beta}} $ being a solution to the Einstein Field Equations with a positive cosmological constant $ \Lambda> 0 $. In this context, $ \scri $ is a three-dimensional space-like hypersurface with metric\footnote{To be considered $ C^\infty $ when necessary for the application of \cref{thm:classic} and \cref{thm:corollary-levi-civita}.} $ \ms{_{ab}} $ ---that transforms as \cref{eq:gauge-ms} under gauge changes---  and normal $ \ct{N}{_{\alpha}}\defeq\cd{_{\alpha}}\Omega $. In general, $ \scri $ has a past and a future component that are disconnected. All the results in this paper apply equally to both of them, hence no distinction will be made.\\
		
		 Under this set-up,	it is well known that the Weyl tensor\footnote{The conventions for the curvature tensor are $ \ct{R}{_{\alpha\beta\gamma}^\mu}\ct{v}{_{\mu}}=\prn{\cd{_{\alpha}}\cd{_{\beta}}-\cd{_\beta}\cd{_{\alpha}}}\ct{v}{_{\gamma}} $ and $ \ct{R}{_{\alpha\beta}}=\ct{R}{_{\alpha\mu\beta}^\mu} $.} $ \ct{C}{_{\alpha\beta\gamma}^\delta} $ vanishes at $ \scri $ and that the rescaled version
			\begin{equation}
				\ct{d}{_{\alpha\beta\gamma}^\delta}\defeq \Omega^{-1}\ct{C}{_{\alpha\beta\gamma}^\delta}
			\end{equation}
		is regular and in general different from zero there.  The Bel-Robinson tensor \cite{Bel1958} $ \ct{\T}{_{\alpha\beta\gamma\delta}} $ vanishes at $ \scri $ too, so one constructs a rescaled version as
			\begin{equation}
				\ct{\D}{_{\alpha\beta\gamma\delta}}=\ct{d}{_{\alpha\mu\gamma}^\nu}\ct{d}{_{\delta\nu\beta}^\mu}+\ctr{^*}{d}{_{\alpha\mu\gamma}^\nu}\ctr{^*}{d}{_{\delta\nu\beta}^\mu}\ ,
			\end{equation}
		where the star denotes the hodge dual operation on the first couple of covariant indices. With it, one defines the asymptotic supermomentum as
			\begin{equation}
				\ct{p}{^\alpha}\defeq-\ct{N}{^\mu}\ct{N}{^\nu}\ct{N}{^\rho}\ct{\D}{_{\mu\nu\rho}^\alpha}\ .
			\end{equation}
		Based on this object, a new way of characterising gravitational radiation in full General Relativity at $ \scri $ in the presence of a non-negative $ \Lambda $ was put forward in \cite{Fernandez-Alvarez_Senovilla2020a,Fernandez-AlvarezSenovilla2020b} and later developed in \cite{Fernandez-AlvarezSenovilla2022a,Fernandez-AlvarezSenovilla2022b} ---for a review, see \cite{Senovilla2022}. For $ \Lambda> 0 $, $ \ct{N}{^\alpha} $ is future-pointing and timelike in a neighbourhood of $ \scri $, and there one can work instead with the unit vector field $ \ct{n}{^\alpha}\defeq N^{-1}  \ct{N}{^\alpha}$, with $ -N^2=\ct{N}{_{\mu}}\ct{N}{^\mu} $. All the information of the rescaled Weyl tensor at $ \scri $ is contained in its electric and magnetic parts defined with respect to this vector field,
			\begin{align}
				\ct{D}{_{ab}}&\defeqs\ct{e}{^\alpha_{a}}\ct{e}{^\beta_{b}} \ct{n}{^\mu}\ct{n}{^\nu}\ct{d}{_{\mu\alpha\nu\beta}}\ ,\label{eq:electric}\\
				\ct{C}{_{ab}}&\defeqs\ct{e}{^\alpha_{a}}\ct{e}{^\beta_{b}} \ct{n}{^\mu}\ct{n}{^\nu}\ctr{^*}{d}{_{\mu\alpha\nu\beta}}\ ,
			\end{align}
			where $ \cbrkt{\ct{e}{^\alpha_{a}}} $ is any set of linearly independent vector fields at $ \scri $, orthogonal to $ \ct{n}{_{\alpha}} $.
			One can also introduce a `canonical' version of the asymptotic supermomentum,
			\begin{equation}
				\ct{\P}{^\alpha}=-\ct{n}{^\mu}\ct{n}{^\nu}\ct{n}{^\rho}\ct{\D}{_{\mu\nu\rho}^\alpha} ,
			\end{equation}
			 and write the decomposition
			\begin{equation}
				\ct{\P}{^\alpha}\eqs \ct{N}{^\alpha}\W+\ct{e}{^\alpha_{a}}\cts{\P}{^a}.
			\end{equation}
		 The tangent part $ \cts{\P}{^a} $ to $ \scri $ is called asymptotic super-Poynting vector field, and the scalar $ \W $, asymptotic superenergy density\footnote{Supernergy formalism shares many properties with electromagnetism \cite{Maartens1998}. For a general treatment, see \cite{Senovilla2000}. For the precise content used in the study of gravitational radiation at infinity, see section 2 of \cite{Fernandez-AlvarezSenovilla2022b} ---and references therein for other applications too.}.  Using $ \cts{\P}{^a} $, there is a way of determining the presence of gravitational radiation at $ \scri $ in full General Relativity that we recall here\footnote{The criterion has an alternative, more geometric description in terms of principal directions of the rescaled Weyl tensor that is not going to be used here. For more details on this, see \cite{Fernandez-AlvarezSenovilla2022b,Senovilla2022}.}:
			\begin{crit}[Asymptotic gravitational-radiation condition with $ \Lambda>0 $]\label{criterionGlobal}
					Consider a three-dimensional open connected subset $ \Delta \subset \scri $.
					There is no radiation on $ \Delta $  if and only if the asymptotic super-Poynting vanishes there 
					\begin{equation*}
			 \cts{\Pc}{^a}\eqsopen 0 \iff \text{No gravitational radiation on }\Delta .
					\end{equation*}
			\end{crit}
		\begin{remark}
			The radiation condition is equivalent stated in terms of the commutator of the electric and magnetic parts of the rescaled Weyl tensor,
				\begin{equation}\label{eq:commutator}
					\commute{D}{C}_{ab}=0 \iff \cts{\P}{^a}=0.
				\end{equation}
		\end{remark}
		Importantly, \cref{criterionGlobal} is invariant under conformal gauge transformations \eqref{eq:gauge-ms}. \\
		
		This characterisation is the first clue in the itinerary towards the isolation of the radiative degrees of freedom of the gravitational field; the second one is the results by Friedrich  \cite{Friedrich1986a,Friedrich1986b}, by means of which a three-dimensional Riemannian manifold endowed with a conformal class of metrics and a traceless and divergence-free (TT) tensor constitute initial/final data that determine a solution of the $ \Lambda $-vacuum Einstein field equations. The Riemannian manifold turns out to be the conformal boundary, and the TT-tensor, the electric part \eqref{eq:electric} of the rescaled Weyl tensor. Hence, all the available information about gravitational radiation has to be encoded in the triplet $ \prn{\scri,\ms{_{ab}},\ct{D}{_{ab}}} $. An additional insight comes form calculating the magnetic part of $ \ct{d}{_{\alpha\beta\gamma}^\delta} $ at $ \scri $, which shows that
			\begin{equation}\label{eq:cotton-york}
				\ct{C}{_{ab}}=-\sqrt{\frac{3}{\Lambda}}\ct{Y}{_{ab}}\ ,
			\end{equation}	
		with $ \ct{Y}{_{ab}} $ being the Cotton-York tensor of $ \prn{\scri,\ms{_{ab}}} $. This proves that $ \ct{C}{_{ab}} $ is completely determined by the curvature of  $ \prn{\scri,\ms{_{ab}}} $, as
			\begin{equation}
				\ct{Y}{_{ab}}=-\frac{1}{2}\ct{\epsilon}{_{a}^{cd}}\cds{_{[c}}\cts{S}{_{d]b}}
			\end{equation}
		where the volume form $ \ct{\epsilon}{_{abc}} $ of $ \scri $ and the intrinsic Schouten tensor $ \cts{S}{_{ab}} $ have been introduced. Therefore, gravitational radiation is determined by an interplay ---\cref{eq:commutator}--- between the intrinsic geometry of $ \prn{\scri,\ms{_{ab}}} $ and an extra piece of information ---the TT-tensor $ \ct{D}{_{ab}} $. This is the perspective adopted here and largely illustrated in \cite{Fernandez-AlvarezSenovilla2022b}.\\
		
		Yet another point that deserves attention is the existence of a `first component of news' $ \ctcn{V}{_{AB}} $ associated to a foliation given by some $ \ct{m}{^{a}} $. It has been shown in \cite{Fernandez-AlvarezSenovilla2022b} that this object must be part of any possible news tensor.  It is convenient to recall some definitions and results from that work.	
			\begin{deff}[Equipped $ \scri $]\label{def:additional-structure}
		 		An open, connected, subset $ \Delta\subset \scri $ with the same topology than $ \scri $ is said to be {\em equipped}  when it is endowed with a congruence $ \C $ of curves characterised by a unit vector field $ \ct{m}{^a} $. The projected surface $ \Scn\defeq \Delta/\C $ and $ \C $ are characterised by the conformal family of pairs
		 		\begin{equation}\label{eq:pairs-additional-structure}
		 		\prn{\ctcn{P}{_{ab}},\ct{m}{_a}}\spacef,
		 		\end{equation}
		 		where $ \ctcn{P}{_{ab}} $ is the projector to $ \Scn $. Two members belong to the same family if and only if $ \prn{\ctcn{P'}{_{ab}},\ct{m'}{_a}}=\prn{\Psi^2\ctcn{P}{_{ab}},\Psi\ct{m}{_a}} $, where $ \Psi $ is a positive function on $ \scri $.
		 	\end{deff}
		 	\begin{deff}[Strictly equipped $ \scri $]\label{def:additional-structure1}
 			 		One says  that $ \scri $ is {\em strictly equipped} when it is equipped and the unit vector field $ \ct{m}{^a} $ is surface-orthognal, providing a foliation by cuts.\footnote{By `cut' one means any two-dimensional Riemannian manifold with induced metric $ \mc{_{AB}} $.}
 			\end{deff}
 		When considering a foliation ($ \ctcn{\omega}{_{ab}}=0 $), one can always write
			\begin{equation}\label{eq:m-congruences-sec}
			\ct{m}{_a}=F\cds{_a}v\quad\text{with }	\frac{1}{F}=\lied_{_{\vec{m}}}v,
			\end{equation}
		where each leaf $ \Sc_{v}$ is labelled by a constant value of the parameter $ v $ along the curves.  Next, define $ \ctcn{S}{_{AB}}\defeq\ctcn{E}{^a_{A}}\ctcn{E}{^b_{B}}\cts{S}{_{ab}} $ and introduce the following combination --- where $ \ctcn{\Sigma}{}^2\defeq\ctcn{\Sigma}{_{AB}}\ctcn{\Sigma}{^{AB}} $ and $ \ctcn{\omega}{}^2\defeq\ctcn{\omega}{_{AB}}\ctcn{\omega}{^{AB}} $:
			\begin{align}
				\ctcn{U}{_{AB}}&\defeq \ctcn{S}{_{AB}}+\frac{1}{2}\cscn{\kappa}\ctcn{\Sigma}{_{AB}}+\ct{L}{_{AB}}\spacef,\label{eq:UABdef-congruences}\\
				\ct{L}{_{AB}}&\defeq \prn{\frac{1}{8}\cscn{\kappa}^2-\frac{1}{4}\cscn{\Sigma}^2+\frac{3}{4}\cscn{\omega}^2}\mcn{_{AB}}\spacef.\label{eq:LAB-congruences}
			\end{align}
		Two key results read as follows:
			\begin{corollary}[The tensor field $ \rho $ for strictly equipped $\scri$ with $\mathbb{S}^2$ leaves]\label{thm:rho-tensor-foliation}
				Assume $\scri$ is strictly equipped with $ v $ the parameter along the curves and such that \cref{eq:m-congruences-sec} holds. If the leaves have $ \mathbb{S}^2 $-topology, there is a unique tensor field $ \ctcn{\rho}{_{ab}} $  on $ \scri $ orthogonal to $ \ct{m}{^a} $ (equivalently, a one-parameter family of symmetric tensor fields $ \ctcn{\rho}{_{AB}}\prn{v}\defeq \ctcn{E}{^a_A}\ctcn{E}{^b_B}\ctcn{\rho}{_{ab}} $ on the projected surface $ \Scn $) whose behaviour under conformal rescalings \eqref{eq:gauge-mcn} is
					\begin{equation}\label{eq:gauge-behaviour-appropriate-foliations}
							\ctcng{\rho}{_{AB}}=\ctcn{\rho}{_{AB}}-a\frac{1}{\omega}\cdcn{_A}\ctcn{\omega}{_B}+\frac{2a}{\omega^2}\ctcn{\omega}{_A}\ctcn{\omega}{_B}-\frac{a}{2\omega^2}\ctcn{\omega}{_C}\ctcn{\omega}{^C}\mcn{_{AB}}\ ,
					\end{equation}
				where $ \ctcn{\omega}{_{A}}\defeq\cdcn{_{A}}\omega $, $ a\in\mathbb{R} $, and satisfies the equation
				 	\begin{equation}\label{eq:rho-diff-eq-foliation}
					\ctcn{P}{^d_a}\ctcn{P}{^e_b}\ctcn{P}{^f_c}\cds{_{[f}}\ctcn{\rho}{_{d]e}}= 0
					\end{equation}
				in any conformal frame. This tensor field must have a trace $ \ctcn{\rho}{^e_e}\defeq \ctcn{P}{^{ae}}\ctcn{\rho}{_{ae}}=a\cscn{K} $ and reduces, at each leaf, to the corresponding tensor $ \ct{\rho}{_{AB}} $ with all its properties. 
			In particular, it is given for the round-sphere one-parameter family of metrics by $ \ctcn{\rho}{_{ab}}=\ctcn{P}{_{ab}}a\cscn{K}/2 $.
			\end{corollary}
		Here, $ \cscn{K} $ is the family of Gaussian curvatures on $ \Scn $ ---see appendix A in \cite{Fernandez-AlvarezSenovilla2022b}.
			\begin{prop}[The first component of news on strictly equipped $\scri$ with $\mathbb{S}^2$ leaves]\label{thm:onepiece-news-foliations}
			Assume $\scri$ is strictly equipped	with $ v $ the parameter along the curves and such that \cref{eq:m-congruences-sec} holds. If the leaves have $ \mathbb{S}^2 $-topology, there is a one-parameter family of symmetric traceless gauge-invariant tensor fields
			\begin{equation}
			\ctcn{V}{_{AB}}\defeq \ctcn{U}{_{AB}}-\ctcn{\rho}{_{AB}}\spacef,
			\end{equation}
			that satisfies the gauge-invariant equation
			\begin{equation}\label{eq:diffUAB-foliations}
			\cdcn{_{[A}}\ctcn{U}{_{B]C}}=\cdcn{_{[A}}\ctcn{V}{_{B]C}}\spacef,
			\end{equation}
			where $ \ctcn{\rho}{_{AB}} $ is the family of tensor fields of \cref{thm:rho-tensor-foliation} (for $ a=1 $). Besides, $ \ctcn{V}{_{AB}} $ is unique with these properties.
		\end{prop}
						\begin{remark}\label{rmk:other-topologies}
							In case the topology of the projected surface $ \Scn $ is not $ \mathbb{S}^2 $, one has  the generalisations of  \cref{thm:rho-tensor-foliation} and \cref{thm:onepiece-news-foliations} ---corollary 6.3 and 6.4 in \cite{Fernandez-AlvarezSenovilla2022b}, respectively. Standard topologies such as $ \mathbb{S}^2 $, $ \mathbb{R}\times\mathbb{S}^1 $ or $ \mathbb{R}^2 $ are permitted.
						\end{remark}		
	In addition, for algebraically special rescaled Weyl tensor at $ \scri $, and  under appropriate conditions, $ \ctcn{V}{_{AB}} $ determines the presence of gravitational radiation (i.e., $ \ctcn{V}{_{AB}}=0\iff\cts{\P}{^a}=0 $, and $ \ctcn{V}{_{AB}} $ is the `whole' news tensor, see Theorem 2 in that reference) ---as an example of this, the C-metric is discussed in \cref{ssec:special-case}. It is worth noting that the news tensor of $ \Lambda=0 $ conformal infinity enters in the expression of the (Bondi-Trautman) energy-momentum derived by Geroch \cite{Geroch1977} and angular momentum using the symplectic formalism \cite{Ashtekar1981c} ---see also \cite{Ashtekar2018}. Hence, any deeper understanding of this class of tensor fields is of interest if one aims at a mass-loss formula in the scenario with $ \Lambda>0 $.\\
		
		After this brief review, observe that if one knows the Levi-Civita connection $ \cds{_{a}} $ of $ \prn{\scri,\ms{_{ab}}} $, then  $ \ct{C}{_{ab}} $ can be computed. Thus, from the viewpoint of \cref{criterionGlobal}, $ \cds{_{a}} $ must contain one part of the radiative degrees of freedom. In fact, one can be tempted to say that it should be characterised by \textit{half of the dof}, noticing the symmetry between the role played by $ \ct{C}{_{ab}} $ and $ \ct{D}{_{ab}} $ in \cref{eq:commutator}. Next, consider the affine space $ \Xi $ of differential operators $ \cdt{_{a}} $ on $ \prn{\scri,\ms{_{ab}}} $. By \cref{thm:corollary-levi-civita}, on each neighbourhood $ \U_{p} $, the Levi-Civita connection belongs to this space, $ \cds{_{a}}\in\Xi $, and it is determined with respect to a fixed origin $ \cdto{_{a}}\in\Xi $ by two functions encoded in $ \cto{s}{^c_{ab}} $, 
			\begin{equation}
				\brkt{\cds{_{a}}-\cdto{_{a}}}\ct{v}{_{b}}=\cto{s}{^e_{ab}}\ct{v}{_{e}}\ .
			\end{equation}
		\emph{These represent one part of the asymptotic radiative degrees of freedom of the gravitational field in full general relativity with a positive cosmological constant}. Note that $ \cto{s}{^a_{bc}} $ corresponds to $ \ct{d\prn{\cds{},\cdto{}}}{^a_{bc}} $ and, complementarily, $ \cds{_{a}} $ can be given coordinates $ \prn{0,0} $ in $ \Xi $ ---see \cref{fig:xi-space}. Remarkably, there is a broad class of space-times for which these two degrees of freedom suffice to determine the content of gravitational radiation at infinity; those having an algebraically special $ \ct{d}{_{\alpha\beta\gamma}^\delta} $ at $ \scri $ ---see \cref{ssec:special-case}.	Noting that $ \ct{t}{^a_{bc}}=0 $ imposes that each element of $ \ct{e}{^a_{i}}\in \Upsilon_{s} $ is surface-orthogonal (i.e., defines a foliation), \cref{thm:onepiece-news-foliations} is applicable to each of the elements $ \cbrkt{\ct{e}{^a_{i}}}=\cbrkt{\ct{m}{^a},\cta{m}{^a},\cth{m}{^a}} $ of the triad, leading to the following result
				\begin{lemma}\label{thm:news-triads}
				 	Let $ \cbrkt{\ct{m}{^a},\cta{m}{^a},\cth{m}{^a}}  $ be the elements of a triad $ \cbrkt{\ct{e}{^a_{i}}}\in \Upsilon_{s} $. Then, provided that the corresponding projected surfaces $ \Scn $, $  \acute{\Scn} $ and $ \hat{\Scn} $ associated with each of the elements satisfy the topology restriction of \cref{thm:onepiece-news-foliations} ---or its generalisation to other topologies indicated in \cref{rmk:other-topologies}---, there exists a set of three first components of news, $ \ctcn{V}{_{AB}} $, $ \ctcnh{V}{_{\Ah\Bh}} $, $ \ctcna{V}{_{\Aa\Ba}} $, corresponding to each of the elements of the triad.
				\end{lemma}
				\begin{remark}
				   Recalling \cref{rmk:sset-notempty} and using \cref{thm:news-triads}, it follows that locally, and provided that the topology requirements of \cref{rmk:other-topologies} are met, it is always possible to find three first components of news.
				\end{remark}
			One further justification of $ \Xi $ containing part of the radiative degrees of freedom comes from \cref{eq:nabla-D-D}. Given a foliation by umbilical cuts ---$ \ctcn{\Sigma}{_{AB}}=0 $--- defined by $ \ct{m}{^a} $, and the set of all possible triads containing this vector field, there is an associated set of operators $ \cbrkt{\cdt{_{a}}} $ that induce the covariant derivative $ \cdcn{_{A}} $ in $ \Scn $. And this determines the associated $ \ctcn{V}{_{AB}} $. A different umbilical foliation, given by $ \ct{m'}{^a} $, has a different $ \ctcn{V'}{_{AB}} $, and its connection $ \cdcn{_{A}}' $ is determined now by a different set of operators $ \cbrkt{\cdtp{_{a}}} $. From the point of view of $ \Xi $, any element of $ \cbrkt{\cdt{_{a}}} $ is determined  from any element of $ \cbrkt{\cdtp{_{a}}} $ by two functions. 
		\subsection{The algebraically special case}\label{ssec:special-case}
			The case with $ \ct{d}{_{\alpha\beta\gamma}^\delta} $ algebraically special at $ \scri $ ---i.e., when the rescaled Weyl tensor has at least one repeated principal null direction at $ \scri $--- is worthy of attention. A result\footnote{The wording has been adapted to suit better the present work.}  of \cite{Fernandez-AlvarezSenovilla2022b} (Lemma 6.8)  is applicable, and allows to translate this condition to an equation involving just $ \prn{\scri,\ms{_{ab}},\ct{D}{_{ab}}} $ and a vector field $ \ct{m}{^{a}} $,
				\begin{lemma}\label{thm:nircCD}
					Assume that $ \ct{d}{_{\alpha\beta\gamma}^\delta} $ is algebraically special at $ \scri $ and denote by $ \ct{m}{^{a}} $ the unit vector field aligned with the projection to $ \scri $ of (any of the) repeated principal null directions. Then
						\begin{equation}\label{eq:nirc}
							\ct{D}{_{ab}}-\frac{1}{2}\ct{D}{_{ef}}\ct{m}{^e} \ct{m}{^f} \prn{3\ct{m}{_a} \ct{m}{_b} - \ms{_{ab}} }= \ct{m}{^d}\ct{\epsilon}{_{ed(a}} \prn{\ct{C}{_{b)}^e} +\ct{m}{_{b)} }\ct{m}{^f} \ct{C}{_f^e}}\spacef.
						\end{equation}
				\end{lemma}
			Thus, in this case, the TT-tensor $ \ct{D}{_{ab}} $ is determined by $ \ct{C}{_{ab}} $ except for the component $ \ct{m}{^{a}}\ct{m}{^b}\ct{D}{_{ab}} $. Nevertheless, the latter is not involved in the gravitational condition, as one has that \cref{eq:nirc} implies 
				\begin{equation}\label{eq:rad-cond-special-case}
					\cts{\P}{^a}=0\iff  \ct{D}{_{ab}}=\frac{1}{2}\ct{D}{_{ef}}\ct{m}{^e} \ct{m}{^f} \prn{3\ct{m}{_a} \ct{m}{_b} - \ms{_{ab}} } \iff \ct{C}{_{ab}}=\frac{1}{2}\ct{C}{_{ef}}\ct{m}{^e} \ct{m}{^f} \prn{3\ct{m}{_a} \ct{m}{_b} - \ms{_{ab}} }\ .
				\end{equation} 
			This last statement follows from Corollary 6.7 in \cite{Fernandez-AlvarezSenovilla2022b}, where the radiant superenergy formalism was used\footnote{With equations (2.50, 2.52) and property (c) on page 16 in that work, one recasts the result as \eqref{eq:rad-cond-special-case}}. Notice that, due to the fact that $ \ct{D}{_{ab}} $ and $ \ct{C}{_{ab}} $ are traceless, one has
				\begin{equation}
					\ct{P}{^{ab}}\ct{D}{_{ab}}+\ct{m}{^a}\ct{m}{^b}\ct{D}{_{ab}}=0\ ,\quad \ct{P}{^{ab}}\ct{C}{_{ab}}+\ct{m}{^a}\ct{m}{^b}\ct{C}{_{ab}}=0\ ,
				\end{equation}
			so that, indeed, no $ \ct{m}{^a}\ct{m}{^b}\ct{D}{_{ab}} $ ($ \ct{m}{^a}\ct{m}{^b}\ct{C}{_{ab}} $) component is involved in \cref{eq:rad-cond-special-case}\ . \emph{Hence, in this particular scenario, the presence of radiation is fully ruled out by the Cotton-York tensor of $ \prn{\scri,\ms{_{ab}}} $, that is, by the two degrees of freedom\footnote{In \cite{Ashtekar2014} it is stated that ``The condition $ \ct{C}{_{ab}}=0 $ removes `half the	radiative degrees of freedom' in the gravitational field and, in addition, the gravitational waves it does allow can not carry any of the de Sitter momenta across $ \scri $''. From the perspective of  \cref{criterionGlobal} , $ \ct{C}{_{ab}}=0 $ not only cuts the degrees of freedom by half, but it kills gravitational radiation ---which somehow solves the puzzle of having waves that `can not carry any of the de Sitter momenta across $ \scri $'. Instead, what it is stated in this work is that half of the radiative degrees of freedom are removed by condition \eqref{eq:nirc} in a very different way, as neither $ \ct{C}{_{ab}} $ nor $ \ct{D}{_{ab}} $ vanish in general, so that the presence of gravitational waves at infinity is completely feasible ---and, indeed, that is the case except for $ \ct{d}{_{\alpha\beta\gamma}^\delta} $ having two repeated principal null directions (type D) at $ \scri $ \emph{co-planar} with the normal $ \ct{N}{^\alpha} $ \cite{Fernandez-AlvarezSenovilla2022b}.} of $ \Xi $}. Notably, this situation is in correspondence with what happens in the scenario with $ \Lambda=0 $, where the phase space of connections introduced by Ashtekar \cite{Ashtekar81} at $ \scri $  has precisely 2 degrees of freedom, as the transport equations for the relevant curvature fields along the generators are sourced by the fields themselves\footnote{ This last point also suggests that the missing pair of dof in the general case may appear at $ \scri $ by some kind of time-derivative of the intrinsic fields of a previous space-like hypersurface.}. Working out further details, one can introduce a notion of incoming radiation, and the analogy with $ \Lambda=0 $ goes deeper, as in that case the `electric' part of the curvature is determined by the `magnetic' part except for the Coulomb component --see remark 6.7 in \cite{Fernandez-AlvarezSenovilla2022b} for more details. Hence, this situation may be the closest one --from the point of view of gravitational radiation-- to the asymptotically flat case. The analogy does not come at any price, as condition \eqref{eq:nirc} constrains the radiative degrees of freedom. Thus, it should not be surprising that adapting non-geometric methods of the $ \Lambda=0 $ scenario ---such as coordinate systems à la Bondi or boundary conditions on metric coefficients \cite{Bonga2023, Saw2016, Compere2019}--- may give answers that differ from the general ones obtained in the present approach.\\
		
			Some relevant examples of exact solutions for which $ \ct{d}{_{\alpha\beta\gamma}^\delta} $ is algebraically special at $ \scri $ include the Robinson-Trautman metrics and the C-metric ---both radiative--- and the Kerr-de Sitter like solutions\footnote{The Kerr-de Sitter like solutions are characterised by having a KVF whose associated rescaled Mars-Simon tensor vanishes. In general, they have non-vanishing $ \ct{D}{_{ab}} $ and $ \ct{C}{_{ab}} $ at $ \scri $, and their form is that of \cref{eq:rad-cond-special-case} with $ \ct{m}{^a} $ being the conformal KVF induced by the space-time KVF, therefore they do not have gravitational radiation at $ \scri $. There is a generalisation called \emph{asymptotically} Kerr-de Sitter like, in which the rescaled Mars-Simon tensor is required to vanish at $ \scri $ but not necessarily everywhere; this more broad class can contain gravitational radiation at infinity.} \cite{Mars2016} ---non radiative---, whose content on gravitational radiation was analysed in \cite{Fernandez-AlvarezSenovilla2022b} ---see also \cite{Senovilla2022}. In what follows the Kerr-de Sitter metric and C-metric with $ \Lambda>0 $ are used to illustrate the relation between shears of different triads, serving as examples for \cref{thm:geom-id}, \cref{thm:corollary-levi-civita,thm:news-triads}, and also to give cases of the surjection $ \varphi $ of \cref{rmk:onto}. The basic radiative properties are concisely reviewed and conventions are followed according to section 8 of \cite{Fernandez-AlvarezSenovilla2022b}.
			 \subsubsection*{Kerr-de Sitter}
			 	Kerr-de Sitter is an example of a space-time with conformally flat metric $ \ms{_{ab}} $ at $ \scri $, hence the Cotton-York vanishes and condition \eqref{eq:commutator} is satisfied, meaning that there is no gravitational radiation arriving at infinity. The metric at $ \scri $ reads 
	 			\begin{equation}
	 					h=N^2\df{t}^2+\frac{\sin^2\theta}{1+N^2a^2} \df{\phi}^2-\frac{N^2a\sin^2\theta}{1+N^2a^2}\prn{\df{\phi} \df{t} + \df{t}\df{\phi}}+\frac{1}{1+N^2a^2\cos^2\theta}\df{\theta}^2\quad\ ,
					\end{equation}
				where $ N^2=\Lambda/3 $. Now, consider the following triad $ \cbrkt{\ct{e}{^a_{i}}} $,
					\begin{align}
						\ct{m}{^{a}}&=\frac{1+N^2a^2}{N\prn{1+N^2a^2\cos^2\theta}^{1/2}}\prn{\delta^a_{t}+aN^2\delta^a_\phi}\ ,\\
						\cta{m}{^a}&=\frac{\prn{1+a^2N^2}^{1/2}}{\sin\theta}\delta^a_{\phi}\ ,\\
						\cth{m}{^a}&=\prn{1+a^2N^2\cos^2\theta}^{1/2}\delta^a_{\theta}\ .
					\end{align}
				It has $ \ctcn{\Sigma}{_{\ma\mh}}=\ctcna{\Sigma}{_{\mh m}}=\ctcnh{\Sigma}{_{\ma m}}=0 $, which trivially fulfils \cref{eq:identity-shears} and shows that $ \ct{s}{^a_{bc}}=0 $. Then, according to \cref{thm:operator-D}, the associated $ \cdt{_{a}} $ coincides with the Levi-Civita $ \cds{_{a}} $ connection. Not only that, but one can compute the vorticities to get $ \ctcn{\omega}{_{\ma\mh}}=\ctcna{\omega}{_{\mh m}}=\ctcnh{\omega}{_{\ma m}}=0  $, which is in agreement with \cref{thm:geom-id} and implies that $ \ct{t}{^a_{bc}}=0 $. Thus, $ \cbrkt{\ct{e}{^a_{i}}}\in \Upsilon_{s} $ and, according to \cref{thm:corollary-levi-civita}, not only $ \cdt{_{a}} $ coincides with $ \cds{_{a}} $, but its action on arbitrary tensor fields is completely determined by the triplet of two-dimensional connections associated to this triad:
					\begin{align}
						\cds{_{a}}\ct{v}{_{b}}=\cdt{_{a}}\ct{v}{_{b}}			&=\ctcn{W}{_{a}^A}\ctcn{W}{_{b}^B}\cdcn{_{A}}\ctcn{v}{_{B}}+\cta{m}{_{a}}\cta{m}{_{b}}\cta{m}{^\Ah}\cta{m}{^\Bh}\cdcnh{_{\Ah}}\prn{\ct{m}{^e}\ct{v}{_{e}}\ct{m}{_{\Bh}}}+\cth{m}{_{a}}\cth{m}{_{b}}\cth{m}{^\Aa}\cth{m}{^\Ba}\cdcna{_{\Aa}}\prn{\ct{m}{^e}\ct{v}{_{e}}\ct{m}{_{\Ba}}}\nonumber\\
						&+ \ct{m}{_{a}}\cta{m}{_{b}}\ct{m}{^\Ah}\cta{m}{^\Bh}\cdcnh{_{\Ah}}\ctcn{v}{_{\Bh}}+\ct{m}{_{b}}\cta{m}{_{a}}\ct{m}{^\Bh}\cta{m}{^\Ah}\cdcnh{_{\Ah}}\ctcn{v}{_{\Bh}}+\ct{m}{_{a}}\cth{m}{_{b}}\ct{m}{^\Aa}\cth{m}{^\Ba}\cdcna{_{\Aa}}\ctcn{v}{_{\Ba}}+\ct{m}{_{b}}\cth{m}{_{a}}\ct{m}{^\Ba}\cth{m}{^\Aa}\cdcna{_{\Aa}}\ctcn{v}{_{\Ba}}\nonumber\\
						+&\ct{m}{_{a}}\ct{m}{_{b}}\brkt{\ct{m}{^\Aa}\ct{m}{^\Ba}\cdcna{_{\Aa}}\ctcn{v}{_{\Ba}}+\ct{m}{^\Ah}\ct{m}{^\Bh}\cdcnh{_{\Ah}}\prn{\cta{m}{_{\Bh}}\cta{m}{^e}\ct{v}{_{e}}}}\ .
					\end{align}  
				In passing by, notice that $ \ct{t}{^a_{bc}}=0 $ implies that one can find scalar functions $ u $, $ v $, $ w $, $ A $, $ B $, $ C $, depending on the coordinates $ \cbrkt{t,\phi,\theta} $, such that
					\begin{equation}
					\ct{m}{_{a}}=A\df{u}\ ,\quad \cta{m}{_{a}}=B\df{v}\ ,\quad\cth{m}{_{a}}=C\df{w}\ .
					\end{equation}
				Indeed, one can easily identify
					\begin{equation}
					u=t\ ,	\quad v=\phi-aN^2t\ ,\quad w=\theta
					\end{equation}
				and 
					\begin{equation}
					A=\frac{N\prn{1+a^2N^2\cos^2\theta}^{1/2}}{1+a^2N^2}\ ,\quad B=\frac{\sin\theta}{\prn{1+a^2N^2}^{1/2}}\ ,\quad C=\frac{1}{\prn{1+N^2a^2\cos^2\theta}^{1/2}}\ .
					\end{equation}
				Then,  $ \cbrkt{u,v,w} $ diagonalises the metric
					\begin{equation}
						h=\frac{N^2\prn{1+a^2N^2\cos^2w}}{\prn{1+a^2N^2}^2}\df{u}^2+\frac{\sin^2w}{1+a^2N^2}\df{v}^2+\frac{1}{1+N^2a^2\cos^2w}\df{w}^2\ ,
					\end{equation}
				which also illustrates \cref{rmk:directions}. In addition to this, it is possible to compute \cite{Fernandez-AlvarezSenovilla2022b} the first component of news $ \ctcn{V}{_{AB}} $ associated to this $ \ct{m}{^a} $ and check that
					\begin{equation}
						\ctcn{V}{_{AB}}=0\ .
					\end{equation}
				
				As a different choice of triad, make a rotation around $ \cth{m}{^a} $ to get a new couple of elements,
					\begin{align}
						\ct{m}{^a}&=\frac{1}{N}\delta^a_{t}\ ,\\
						\cta{m}{^a}&=\frac{1}{\prn{1+a^2N^2\cos^2\theta}^{1/2}}\brkt{a\sin\theta\delta^a_{t}+\prn{1+a^2N^2}\delta^a_{\phi}}\ .
					\end{align}
				Observe that $ \ct{m}{^a} $ is a Killing vector field, thus one has $ \ctcn{\Sigma}{_{\ma\mh}}=0  $. For the other two shears one finds
					\begin{equation}
						\ctcna{\Sigma}{_{\mh m}}=-\ctcnh{\Sigma}{_{\ma m}}=aN\cos\theta\ ,
					\end{equation}
				which again satisfies \cref{eq:identity-shears}. That is, in this case $ \ct{s}{^a_{bc}}\neq 0 $ and the operator $ \cdt{_{a}} $ is different from the previous one, so that it does not coincide with $ \cds{_{a}} $. This $ \cdt{_{a}} $ can be labelled with the values of any of the three shears' components, e.g. $ \prn{\ctcn{\Sigma}{_{\ma\mh}},	\ctcna{\Sigma}{_{\mh m}}}=\prn{0,aN\cos\theta} $. 
			 \subsubsection*{C-metric with positive cosmological constant}
			 	The C-metric (here assuming $ \Lambda>0 $) represents two black-holes accelerating \cite{Podolsky2000}. Hence, it is not surprising that it contains gravitational radiation at $ \scri $ \cite{Fernandez-AlvarezSenovilla2020b}, whose metric in the gauge choices of \cite{Fernandez-AlvarezSenovilla2022b} reads 
			 	---see also \cite{Podolsky2023} for a recent treatment of the generalised black holes metrics of Petrov type D which include the C-metric:
			 		\begin{equation}
					h= (a^2S+N^2) \df{\tau^2}+\frac{N^2}{S(a^2S+N^2)}\df{p^2}+S\df{\sigma^2}\quad.
					\end{equation}
				Here $ 2ma< 1 $, $ S(p) \defeq (1-p^2)(1-2a mp)\geq 0 $ with $ p\in(-1,1] $, $ \sigma\in\Big[0,2\pi /(1-2am)\Big) $ and $ N^2=\Lambda/3 $ as before. The rescaled Weyl tensor is algebraically special at $ \scri $ but the metric is \emph{not} conformally flat. Thus \cref{eq:nirc} holds and all the information about gravitational radiation can be encoded in the Cotton-York tensor $ \ct{Y}{_{ab}} $ --see \cref{eq:cotton-york}--, and having a non-vanishing asymptotic super-Poynting vector field,
				\begin{equation}\label{eq:canonical-super-P-C-metric}
				\cts{\Pc}{^a}\eqs \sqrt{\frac{3}{\Lambda}}18a m^2 S  \prn{1+\frac{6}{\Lambda}Sa^2}\delta_{p}^a\ .
				\end{equation}
				 It was shown in \cite{Fernandez-AlvarezSenovilla2022b} that, associated to the foliation given by 
					 \begin{equation}\label{eq:m-cmetric}
	 					\ct{m}{^a}=\prn{-\frac{N}{\csS{T}}\delta_\tau^a-\frac{a S}{N}\delta_p^a}
	 				\end{equation}
				 the first component of news $ \ct{V}{_{ab}} $ reads
		 	 \begin{equation}
				 \ctcn{V}{_{AB}}= \frac{H}{S^2}\cdcn{_A}p\cdcn{_B}p - H \cdcn{_A}\sigma\cdcn{_B}\sigma\ ,
				 \end{equation}	 
			where on each leaf $ \Sc $
			 \begin{equation}
					 H\defeqc \int 3a m S dp\eqc 3a m\prn{\frac{1}{2}a mp^4-\frac{1}{3}p^3-a m p^2+p}+ \frac{3}{2}m^2a^2- 2a m\ .
			 \end{equation}
			 Indeed, not only $ \ct{m}{^a} $, $ \ct{C}{_{ab}} $ and $ \ct{D}{_{ab}} $ obey \cref{eq:nirc}, but also
			 	\begin{equation}
			 		\ctcn{V}{_{AB}}=0 \iff \cts{\P}{^a}=0\ \iff \text{ no gravitational radiation ($ a=0 $)}\ .
			 	\end{equation}
			 One can complete a triad by choosing, for instance, 
			 	\begin{align}
			 		\cta{m}{^a}&=\frac{aS^{1/2}}{Sa^2+N^2}\delta^a_{\tau}+S^{1/2}\delta^a_{p}\ ,\\
			 		\cth{m}{^a}&=\frac{1}{S^{1/2}}\delta^a_{\sigma}\ .
			 	\end{align}
			 The explicit computation of the shears gives $ \ct{s}{^a_{bc}}=0 $. The vorticity of $ \ct{m}{_{a}} $ was computed in \cite{Fernandez-AlvarezSenovilla2022b}, and it vanishes. Thus, using these results and \cref{thm:geom-id} one has also that $ \ct{t}{^a_{bc}}=0 $. That is, $ \cbrkt{\ct{e}{^a_{i}}}\in \Upsilon_{s}  $, so that $ \cdt{_{a}}=\cds{_{a}} $ and this connection is determined by the triplet of two-dimensional derivative operators $ \cbrkt{\cdcn{_{A}},\cdcnh{_{\Ah}},\cdcna{_{\Aa}}} $. In addition, it serves as a non-trivial example of the surjection $ \varphi $ from the space of triads $ \Upsilon $ to the space of connections $ \Xi $ --see \cref{rmk:onto}--, as the triad 
			 	\begin{align}
			 		\ct{m}{_{a}}&=\prn{(a^2S+N^2)}^{1/2}\cds{_{a}}\tau\ ,\\
			 		\cta{m}{_{a}}&=\prn{\frac{N^2}{S(a^2S+N^2)}}^{1/2}\cds{_{a}}p\ ,\\
			 		\cth{m}{_{a}}&=S^{1/2}\cds{_{a}}\sigma,
			 	\end{align}
			 which diagonalises the metric in the given coordinate system, is also in $ \Upsilon_{s} $ and produces the same $ \cdt{_{a}} $.
			  \\
			 
			 It is worth showing that a different triad can be defined, such that $ \ct{s}{^a_{bc}}\neq 0\neq \ct{t}{^a_{bc}} $. Keeping the $ \ct{m}{^{a}} $ of \cref{eq:m-cmetric}, now choose
			 	\begin{align}
			 		\cth{m}{^a}&=\frac{a}{N^2+a^2S}\delta^a_{\tau}+\delta^a_{p}-\frac{\sqrt{S-1}}{S}\delta^a_{\sigma}\ ,\\
			 		\cta{m}{^a}&=\frac{\sqrt{S-1}a}{a^2S+N^2}\delta^a_{\tau}+\sqrt{S-1}\delta^a_{p}+\frac{1}{S}\delta^a_{\sigma}\ .
			 	\end{align}
			 This triad is only defined for $ S>1 $, but for fixed $ a $ and $ m $ (obeying $ 2am<1 $) one can always find a region on $ \scri $ (a range of values for de coordinate $ p $) in which this condition is satisfied; enough for the purposes of this example. The new components of the shears read
			 	\begin{equation}
			 	 	\ctcna{\Sigma}{_{\mh m}}=-\ctcnh{\Sigma}{_{\ma m}}=-\frac{a}{2N\sqrt{S-1}}\dpart{p}S .
			 	\end{equation}
			 This leads to a different connection $ \cdt{_{a}}\neq\cds{_{a}} $, and can be given coordinates $ \prn{\ctcn{\Sigma}{_{\ma\mh}},	\ctcnh{\Sigma}{_{\mh m}}}=\prn{0,a\prn{2N\sqrt{S-1}}^{-1}\dpart{p}S } $.
		\subsection{Comparison with Ashtekar's phase space}
			 The investigation by Ashtekar \cite{Ashtekar81} of the radiative degrees of freedom in asymptotically flat space-times has ground differences with respect to the one put forward here. One of them, which underlies everything else, is the lightlike nature of $ \scri $ when $ \Lambda=0 $. Yet, inspired by Ashtekar's work, the present work aims at isolating the asymptotic radiative degrees of freedom by constructing a space of connections. Thus, some similarities arise and the next brief comparison between Ashtekar's phase space\footnote{Elements of Ashtekar's phase space are equivalence classes of connection arising from some restricted conformal rescaling. Here, such identification is unnecessary.} $ \Gamma $ and the space $ \Xi $ can be depicted:\\
						
						At $ \scri $ with $ \Lambda=0 $,
							\begin{enumerate}
								\item Each of the (equivalence classes of) connections in the space $ \Gamma $ determines a curvature $ \ctr{^*}{K}{^{ab}} $ ---$ \ctr{^N}{C}{^{ab}} $ in the notation of \cite{Fernandez-AlvarezSenovilla2022a}--- and a news tensor. When seen from the viewpoint of the space-time, $ \ctr{^*}{K}{^{ab}} $ is the pullback of the `magnetic' part ---w.r.t. $ \ct{N}{^a}$---  of the rescaled Weyl tensor, and the induced connection on $ \scri $ selects one of these connections.
								\item The space $ \Gamma $ has an affine structure and can be parametrised by two functions. These correspond to the 2 degrees of freedom of the radiative components of the gravitational field. In this case, due to the lightlike character of $ \scri $, they represent all the degrees of freedom of the phase space of gravitational radiation. This is supported by the radiation condition based on the vanishing of the news tensor ---or, equivalently, by the vanishing of the  asymptotic (radiant)  supermomentum $ \ct{p}{^\alpha} $  \cite{Fernandez-AlvarezSenovilla2022a}.
								\item To specify the induced connection with respect to a suitably fixed point in $ \Gamma $, it is enough to give the shear of a one-form $ \ct{\ell}{_{b}} $ that satisfies $ \ct{\ell}{_{b}}\ct{N}{^b}=-1 $ ---where small Latin indices are used, as the normal to $ \scri $ is also tangent.
							\end{enumerate}
							
						At $ \scri $ with $ \Lambda > 0 $,
							\begin{enumerate}
							\item On each neighbourhood $ U_{p}\in\scri$, there is a differential operator $ \cts{\D}{_{a}} \in \Xi$ that coincides with the Levi-Civita connection $ \cds{_{a}} $ of $ \prn{\scri,\ms{_{ab}}} $. This determines\footnote{One could also take a different perspective; namely, defining with each $ \cdt{_{a}} $ a curvature and a Cotton-York like tensor.} the Cotton-York tensor $ \ct{Y}{_{ab}} $ of the manifold. Also, from the viewpoint of the space-time, $ \ct{Y}{_{ab}} $ is the magnetic part of the rescaled Weyl tensor at $ \scri $, and $ \cds{_{a}} $, the induced connection ---in this sense, $ \cds{_{a}} $ on each $ \U_{p} $ `selects' a $ \cdt{_{a}}\in \Xi $.
							\item The space $ \Xi $ has an affine structure and can be parametrised by two functions. Due to the space-like character of $ \scri $, they cannot represent the complete set of radiative degrees of freedom in the general case. This is supported by the radiation \cref{criterionGlobal} based on the vanishing of the tangent part of the asymptotic supermomentum $ \ct{p}{^a} $ ---which involves $ \ct{D}{_{ab}} $ too. However, for algebraically special rescaled Weyl tensor at $ \scri $, these 2 degrees of freedom determine completely the gravitational radiation.
							\item To specify the induced connection with respect to an arbitrary point $ \cdt{_{a}} $ in $ \Xi $, it is enough to give the crossed components of two shears of a triad $ \cbrkt{\ct{e}{^a_{i}}} $.
							\end{enumerate} 
	\section{Discussion}
			It has been shown that an affine space $ \Xi $ of differential operators $ \cts{\D}{_{a}} $ emerges in Riemannian  manifolds from the space $ \Upsilon $ of possible triads $ \cbrkt{\ct{e}{^a_{i}}} $, such that the action of each $ \cts{\D}{_{a}} $ on arbitrary tensor fields depends on the action of a triplet of two-dimensional connections $ \cbrkt{\cdcn{_{A}},\cdcna{_{A}},\cdcnh{_{A}}} $ and an antisymmetric term $ \ct{t}{^a_{bc}} $ containing the vorticities of $ \cbrkt{\ct{e}{^a_{i}}} $. Also, that there is a surjective map $ \varphi $ from $ \Upsilon $ to   $ \Xi $. Remarkably, for $ C^\infty $ Riemannian metrics, the Levi-Civita connection belongs to $ \Xi $, and is fully determined by a triplet  $ \cbrkt{\cdcn{_{A}},\cdcna{_{A}},\cdcnh{_{A}}} $; hence, the subset $ \Upsilon_{s}\in\Upsilon $ of triads that have $ \ct{t}{^a_{bc}}=0 $ and are mapped to $ \cds{_{a}} $ is not empty. Apart from that, points in $ \Xi $ can be labelled by 2 functions which have been justified to represent half of the degrees of freedom of the radiative gravitational field at infinity with a positive cosmological constant. Further investigation has to be carried out to extract all the possible structure in $ \Xi $.\\
			
			One way of understanding the spirit of this work comes from the dimension and causal character of $ \scri $. It can be argued that local methods do not suffice to determine the radiative degrees of freedom of gravity ---e.g., see discussion in chapter 9 of \cite{Penrose1986}. Also, it is reasonable to think of gravitational radiation as linked to two-dimensional submanifolds. As an example, the theorems for the existence of news tensors ---either the one of the asymptotically flat case or the first component of news with $ \Lambda>0 $--- depend on two-dimensional cuts  ---e.g., see how the dimension intervenes in the proof of Corollary 5.2 in \cite{Fernandez-AlvarezSenovilla2022b}. Additionally, observables such as the energy-momentum or mass-loss formula involve integration over a two-dimensional surface, which also relates to topology: when the cosmological constant vanishes, $ \scri $ has $ \mathbb{R}\times\mathbb{S}^2 $ topology \cite{NewmanRPAC1989} and the null generators naturally provide the conformal boundary with a $ 1+2 $ decomposition. For $ \Lambda>0 $, however, the topology of $ \scri $ is not unique \cite{Mars2017,Ashtekar2014} and, in general, there is no natural choice of an intrinsic evolution direction equipping infinity with a $ 1+2 $ splitting. The space of triads $ \Upsilon$ represents this directional freedom and induces a space of differential operators $ \Xi $. From it, the geometry of the two-dimensional projected surfaces associated with a triad is used to show that the Levi-Civita connection belongs to $ \Xi $, in a way that the curvature of $ \scri $ can be determined by two-dimensional connections ---see  \cref{rmk:levi-civitta}. Indeed, this result ---\cref{thm:corollary-levi-civita}--- depends critically on the dimension, and generalisations of \cref{thm:classic} to dimension greater than 3 require restrictions on the curvature \cite{Tod1992}. \\
			 
			On a different note, one of the big pieces missing is how to account for the total number of degrees of freedom in the general case ---that is, when $ \ct{d}{_{\alpha\beta\gamma}^\delta} $ is algebraically general at $ \scri $. A possible way out is to construct a map from one detached abstract Riemannian manifold to another, such that the image of $ \ct{Y}{_{ab}} $ on the second manifold is identified with $ \ct{D}{_{ab}} $. The intuition behind this proposal is the following: the electric and magnetic parts of the rescaled Weyl tensor can be expressed as
			 \begin{align}
			 	\ct{C}{_{ab}}&= \sqrt{\frac{3}{\Lambda}}\brkt{\ct{\epsilon}{^{pq}_a}\cds{_{[p}}\ctd{\sigma}{_{q]b}}+\frac{1}{2}\ct{\epsilon}{^p_{ba}}\cds{_c}\ctd{\sigma}{^c_p}}\spacef,\label{eq:magneticShear}\\
				\ct{D}{_{ab}}&=\frac{1}{2}\sqrt{\frac{3}{\Lambda}}\ct{\ddot{\sigma}}{_{ab}}\spacef,\label{eq:electricShear}
			 \end{align}
			 with $ \ct{\sigma}{_{ab}} $ being the shear of $ \ct{n}{^\alpha} $, and where the dot denotes covariant differentiation along this vector field. This suggests the possibility of identifying the missing pair of degrees of freedom as the evaluation at $ \scri $ of a time derivative of the other pair  ---or as a map between detached three-dimensional Riemannian manifolds. Then, the two pairs would represent the 4 degrees of freedom of the \emph{phase space} of gravitational radiation with $ \Lambda>0 $. This matter has to be tackled rigorously elsewhere.\\
			 
			 Yet another important issue is how asymptotic symmetries come into play, and how they preserve/change the structure of $ \Xi $. As an example, distances between points of $ \Xi $ are left invariant by conformal transformations of the Riemannian metric ---in the sense that operators $ \cdt{_{a}} $ change, but the labels $ \ct{s}{^a_{bc}} $ do not---, see \cref{eq:gauge-distance-xi}. Thus, it is to be expected that basic infinitesimal asymptotic symmetries, i.e., those generated by CKVF $ \ct{\xi}{^a} $ of $ \prn{\scri ,\ms{_{ab}}} $ that satisfy
			 	\begin{equation}
					\lied_{\vec{\xi}}\ct{D}{_{cd}}=-\frac{1}{3}\cds{_a}\ct{\xi}{^a} \ct{D}{_{cd}}\spacef,
				\end{equation}
			 act simply and transitively on $ \Xi $. Additionally, it is necessary to understand the interplay between translations, the first piece of news associated with a congruence of curves \cite{Fernandez-AlvarezSenovilla2022b} and $ \Xi $ ---\cref{thm:news-triads} goes in that direction.
			 \\
			 			 
			 The final goal is to identify a complete phase space for the radiative gravitational field at $ \scri $ with $ \Lambda>0 $ that could be used to shed light on open problems, such as the formulation of the total energy-momentum carried away by gravitational waves in an accelerating expanding universe. Hereby, this work is just a first step on that itinerary; other important pieces, as the ones pointed out above, are yet to be fit into the jigsaw. 
			
			\subsection*{Acknowledgments}

					The paper was completed during a one-year stay at Queen Mary University of London. The author acknowledges the hospitality from the Geometry, Analysis and Gravitation research group and, specially, from Juan A. Valiente Kroon. Also, he is very grateful to Iñaki Garay for encouragement and interesting discussions during the early stages of this work, and to  José M. M. Senovilla for valuable comments, including suggestions on the first draft of the manuscript. Work supported under Grant Margarita Salas MARSA22/20 (Spanish Ministry of Universities and European Union), financed by European Union -- Next Generation EU.

\printbibliography
\end{document}